\documentclass{article}
\usepackage{natbib}
\usepackage[utf8]{inputenc}
\usepackage{graphicx}
\usepackage{setspace}
\usepackage{amsmath}
\usepackage{amsfonts}
\usepackage{amsthm}
\usepackage{comment}
\usepackage{pdfpages}
\usepackage{multirow}
\newtheorem{lemma}{Lemma}
\newtheorem{theorem}{Theorem}
\newtheorem{algorithm}{Algorithm}
\newtheorem{corollary}{Corollary}[lemma]
\newtheorem{definition}{Definition}
\newtheorem{conditions}{Conditions}
\definecolor{paletteRed}{RGB}{223,83,107}
\definecolor{paletteGreen}{RGB}{97,208,79}
\definecolor{paletteBlue}{RGB}{34,151,230}

\usepackage{titlesec}
\usepackage{numprint}
\titleformat*{\section}{\large\bfseries}
\titleformat*{\subsection}{\normalsize\bfseries}

\usepackage{geometry}
\title{Simulation-calibration testing for inference in Lasso regressions}
\date{}
\author{Matthieu Pluntz\footnote{Corresponding author: matthieu.pluntz@gmail.com} \footnote{High-Dimensional Biostatistics for Drug Safety and Genomics, CESP, Université Paris-Saclay, UVSQ, Université Paris-Sud, Inserm, Villejuif, France.} , Cyril Dalmasso\footnote{Laboratoire de Mathématiques et Modélisation d'Évry (LaMME), Université d'Évry Val d'Essonne, Évry, France.} , Pascale Tubert-Bitter$^{\dagger}$\footnote{Contributed equally.} , Ismaïl Ahmed$^{\dagger \S}$}

\begin{document}

\maketitle

\abstract{
We propose a test of the significance of a variable appearing on the Lasso path and use it in a procedure for selecting one of the models of the Lasso path, controlling the Family-Wise Error Rate.
\\
Our null hypothesis depends on a set A of already selected variables and states that it contains all the active variables. We focus on the regularization parameter value from which a first variable outside A is selected. As the test statistic, we use this quantity’s conditional p-value, which we define conditional on the non-penalized estimated coefficients of the model restricted to A. We estimate this by simulating outcome vectors and then calibrating them on the observed outcome’s estimated coefficients. We adapt the calibration heuristically to the case of generalized linear models in which it turns into an iterative stochastic procedure. We prove that the test controls the risk of selecting a false positive in linear models, both under the null hypothesis and, under a correlation condition, when A does not contain all active variables.
\\
We assess the performance of our procedure through extensive simulation studies. We also illustrate it in the detection of exposures associated with drug-induced liver injuries in the French pharmacovigilance database.
}

\section{Introduction}
Variable selection in high-dimensional regressions is a classic problem in health data analysis. It aims to identify a limited number of factors associated with a given health event among a large number of candidate variables such as genetic factors or environmental or drug exposures. The Lasso regression \citep{tibshirani_regression_1996} provides a series of sparse models which variables enter one after another as the regularization parameter $\lambda$ decreases. To be used in variable selection, the Lasso requires a further procedure for choosing the value of $\lambda$ and thus the associated model. By itself, the Lasso does not indicate whether the variables it selects are statistically significant.
\\
\cite{lockhart_significance_2014} address this by adapting the paradigm of hypothesis testing to the Lasso. They develop a significance test for $j_k$, the $k$-th variable selected on the Lasso path, by measuring the impact of its inclusion in the model compared to the selection by the Lasso of the set $A$ of the first $k-1$ variables only. Their test statistic is a difference of covariances between the Lasso results on two datasets --- complete or restricted to the variables in $A$ --- at the regularization parameter $\lambda_{k+1}$ where the next variable enters the Lasso path. They determine the distribution of this statistic in the linear model, under the null hypothesis asserting that the first $k-1$ variables selected by the Lasso contain all the active variables. It is therefore possible to calculate the exact p-value associated with this test. However, in the binary model, the test statistic's exact distribution is not known.
\\
We propose a testing procedure that addresses the same issue as the test by Lockhart et al.: produce, for each variable selected by the Lasso, a p-value that measures its significance while correcting from the drawbacks of this test: invalidity in nonlinear cases, and sometimes lack of power. We retain the idea of exploiting the value of the regularization parameter at which a given variable enters the Lasso path. All else being equal, the earlier on the Lasso path a variable is selected --- that is, at a high value of $\lambda$ --- the stronger and more significant its association with the response. Like Lockhart et al., we propose to test in the framework of hypothesis testing, for each of the variables entering the Lasso path, whether its selection $\lambda$ is higher than it should be if there were no link between this variable and the response variable.
\\
To this end, we are inspired by the permutation selection of \cite{sabourin_permutation_2015}. Their method aims to determine from the data an optimal value of $\lambda$, and then select the model given by the Lasso at this value. The idea is that the optimal $\lambda$ is the one below which it is likely that the Lasso inadvertently selects variables that are in reality independent of $Y$. To find it, they generate a population of random permutations $y^{(l)}, l = 1, .., N$ of the observed response $y$. This imposes "independence"  between the permuted responses and each of the $X_j$ (in the sense that all the $y^{(l)}_i$ have the same distribution despite the variation in $X_{i,j}$). Then they perform the Lasso regression of each permuted response vector on $X$, and measure $\lambda_0 (y^{(l)})$, the smallest value of $\lambda$ at which this Lasso regression selects no variable. The chosen regularization parameter is then the median of $\left\{ \lambda_0 (y^{(1)}), .., \lambda_0 (y^{(N)}) \right\}$.
\\
We adopt the idea of generating responses which follow the same distribution as the observed response, but which has been made artificially independent from the covariates. As in permutation selection, we perform the Lasso regression of the simulated response on the observed covariates and we focus on $\lambda_0$, the highest value of $\lambda$ where a covariate, which by construction is in fact independent of the response, is selected. These $\lambda_0$ obtained from several simulated datasets make up a reference population which is representative of the case where the response and covariates are independent. Unlike permutation selection, where one applies the median $\lambda_0$ from this population to the data of interest, we compare the $\lambda_0$ obtained on the data of interest to this reference population in order to test the significance of this $\lambda_0$, estimating the p-value of the test by Monte Carlo from the reference population. This generalizes to quantities other than $\lambda_0$: the $\lambda$ at which any variable enters the Lasso path, not just the first one. While Sabourin et al. retain only the distribution of the simulated response vectors while breaking their association with all the $X_j$ through permutation, we instead retain the simulated responses' association with some of the covariates by applying a “calibration” to them. Therefore our algorithm is called simulation-calibration.
\\
From a variable selection perspective, if the test concludes that the $\lambda$ at which a variable enters the Lasso path is significant, we select the variable. By iterating the test, we obtain a complete model selection procedure.
\\
Section \ref{problem statement simcal} presents our null hypothesis and basic notations. Section \ref{conditional pvalue} shows the need for, and defines, the \textit{conditional p-value} which we which we will use as a test statistic in the linear case. Section \ref{section linear simcal algorithm} presents the simulation-calibration algorithm estimating this statistic, and section \ref{prop simcal} shows that it is valid and permits control of the type 1 error. Section \ref{section nonlinear} adapts this test and this algorithm to generalized linear models. Section \ref{section proc simcal} presents iterative model selection procedures based on simulation-calibration. Section \ref{extended theorem} states a theorem which permits the model selection procedures to control the FWER (or FDR) under certain conditions in the linear case. Section \ref{simulation studies} shows extended simulation studies of the validity of the simulation-calibration test and of the iterative procedures' variable selection performances. Section \ref{BNPV} shows an application to real-word pharmacovigilance data. 

\section{Problem statement and notations}\label{problem statement simcal}

Let $X \in \mathbb{R}^{n \times p}$ be a matrix of $p$ covariates and $n$ observations and $y \in \mathbb{R}^{n}$ a response vector of size $n$. We consider the linear regression model:

\begin{equation} \label{linmodel}
y = \beta_0 + X\beta + \epsilon, \ \ \ \ \epsilon \sim \mathcal{N} \left(0,\sigma^2 I_n \right)
\end{equation}
\\
where $\beta_0 \in \mathbb{R}$, $\beta \in \mathbb{R}^p$, and $\sigma > 0$. Let $\hat{\beta}^{Lasso} (\lambda,y)$ be the Lasso estimator of $\beta$, which depends on the regularization parameter $\lambda$ and the response $y$. It also depends on the covariate matrix $X$, but this is considered constant in the following.
\\
Let $A$ be a subset of ${1, .., p}$. We assume that the covariates $X_j, j \in A$, are active, and we seek to determine if other covariates are active. We thus define the null and alternative hypotheses, which depend on $A$:
\begin{eqnarray}
H_0 (A) : \forall j \notin A,\beta_j = 0 \nonumber \\
H_1 (A) : \exists j \notin A,\beta_j \neq 0 \nonumber
\end{eqnarray}
\\
Although $A$ is \textit{a priori} an arbitrary subset of $\{1, .., p\}$, the idea, as in Lockhart et al., is to take $A$ as the set of indices of the first $k-1$ covariates entering the Lasso path, for relatively small values of $k$ (even if $p$ can be large). The simplest particular case is $A = \emptyset$.
\\
We consider the following statistic:
$$\lambda_A(y) = \sup \{ \lambda \geq 0 : \ \exists j \notin A, \ \hat{\beta_j}^{Lasso} (\lambda,y) \neq 0 \} $$
\\
That is, the $\lambda$ at which the first variable outside of $A$ is selected, or the $k$-th variable selected if $A$ includes the first $k-1$ variables. We will reject the null hypothesis if we observe a value of $\lambda_A$ that is too high compared to what is expected under the null hypothesis. An abnormally high $\lambda$ at which a variable enters the Lasso typically results from an association between the response and this variable. Therefore, we interpret a test of $H_0 (A)$ based on $\lambda_A$ as assessing the significance of the first variable selected by the Lasso outside of $A$, whose index we call $j_A$. 
\\
However, $H_0 (A)$ might be false while the variable $j_A$ is inactive, if there exists an active variable not belonging to $A$ that is selected "later" on the Lasso path, at a $\lambda$ lower than $\lambda_A$. In this case, it is correct to reject $H_0 (A)$ from a hypothesis testing perspective, but incorrect to conclude that the variable $j_A$ is significant. In section \ref{extended theorem}, we show under which conditions the probability of this event is controlled.
\\
Furthermore, it is possible that multiple variables enter the Lasso path simultaneously at the parameter $\lambda_A$, whether mathematically (for example, in the case of both $X$ and $y$ being binary, there is a non-zero probability of exact symmetry between the associations of $y$ with two distinct covariates) or approximately if the difference between the $\lambda$ at which two distinct variables are selected is too small to measure. In both cases, rejecting $H_0 (A)$ leads to the selection of all these variables.

\section{Conditional p-value} \label{conditional pvalue}

If $\lambda_A$ were used as the test statistic, its associated p-value would be, by definition:
$$p_A^0 (y) = \mathrm{P}_{H_0 (A)} \left( \lambda_A (Y) \geq  \lambda_A (y) \right)$$

where $Y$ is a random variable that follows the model \ref{linmodel}.

In general, there is no simple explicit analytical expression for the result of a Lasso regression. This prevents determining the exact distribution of $\lambda_A (y)$ and thus calculating $p_A^0 (y)$ or other similar quantities. Therefore, we aim to estimate this probability using the Monte Carlo method. To do this, we need to simulate i.i.d. outcome vectors which follow model \ref{linmodel} and satisfy $H_0(A)$.

However, this not possible without additional information because even assuming $H_0(A)$ is true, the parameters of model \ref{linmodel} are not known: $H_0(A)$ guarantees that $\beta_j = 0$ for $j \notin A$ but says nothing about $\beta_j$ for $j \in A$ or $\sigma$. A naive attempt to perform this simulation, where $\beta_0$, the $\beta_j$ for $j \in A$, and $\sigma$ are simply replaced by their estimates on the data, leads to a biased estimation of the p-value (see section \ref{etude simu simcal H0} and the supplementary material).

Therefore, instead of $p_A^0 (y)$, we focus on a variant whose definition includes all the necessary information for its estimation by Monte Carlo. It requires conditioning on the parameters that remain unknown under $H_0(A)$. Consider the following linear model:

\begin{equation} \label{modelA}
y = X_A \beta_A + \epsilon_A, \ \ \ \ \epsilon_A \sim \mathcal{N}(0,\sigma_A^2 I_n)
\end{equation}

where $X_A$ is the matrix composed of a column vector of $1$s and the columns of $X$ whose indices belong to $A$, and $\beta_A$ is the vector $(\beta_j)_{j \in \{0\} \cup A}$. 

This is the reduced form that the model \ref{linmodel} takes when $H_0 (A)$ is true. Then we have $\sigma = \sigma_A$. Let $\theta_A = (\beta_A,\sigma_A)$ be the parameter vector of the model \ref{modelA}, $\Theta_A = \mathbb{R}^{1+|A|} \times \mathbb{R}_{+}$ the space of its possible values, and $\widehat{\theta_A} (y)$ the maximum likelihood estimator (unpenalized) of $\theta_A$. We define the following test statistic:
\begin{definition}[$p_A$ statistic] \label{pAy}
$$p_A (y) = \mathrm{P}_{H_0 (A)} \left(\lambda_A (Y) \geq  \lambda_A (y) | \widehat{\theta_A} (Y) = \widehat{\theta_A} (y) \right).$$
\end{definition}

This formula is similar to that of a p-value but with an added conditioning, making $p_A (y)$ a \textit{conditional p-value}.
\\
$p_A (y)$ is defined so that we can know its distribution as precisely as possible. We have the following result:

\begin{lemma} \label{inequality pAY}
Under $H_0 (A)$,
$$\forall t \in [0,1], \mathrm{P}_{H_0 (A)} \left(p_A(Y) \leq t \right) \leq t.$$
\end{lemma}

In other words, the null distribution of $p_A(Y)$ \textit{stochastically dominates} the uniform distribution on $[0,1]$. This phrase, used in decision theory, means that the cumulative distribution function of the dominant probability distribution is lower at every point than that of the dominated distribution, and thus the dominant distribution is systematically shifted towards higher values compared to the dominated distribution.
\\
In practice, the distribution of $\lambda_A (Y)$ conditional on $\widehat{\theta_A} (Y) = \theta_{A}'$ is typically continuous for all parameter vectors with a non-zero standard deviation, i.e., any $\theta_{A}' \in \Theta_A^*$ where $\Theta_A^* = \mathbb{R}^{|A|} \times \mathbb{R}_{+}^*$. We call the assumption that this is true the \textit{continuity assumption}. It implies the following simpler result:

\begin{corollary} \label{uniform distribution pAY}
Under the continuity assumption, $p_A(Y)$ follows a uniform distribution on $[0,1]$.
\end{corollary}
The proofs of Lemma \ref{inequality pAY} and corollary \ref{uniform distribution pAY} are in the supplementary material.

Lemma \ref{inequality pAY} allows us to use $p_A (y)$ as a test statistic. This quantity is close to $0$ when $\lambda_A (y)$ is high, which is indicative of $H_1 (A)$. Therefore, we wish to reject $H_0 (A)$ if $p_A (y)$ is sufficiently small. The lemma ensures that for all $\alpha \in [0,1]$, rejecting $H_0 (A)$ if and only if $p_A (y) \leq \alpha$ guarantees that the type 1 error is at most $\alpha$. The p-value associated with $p_A (y)$ is $p_A(y)$ itself under the continuity assumption, and not larger than $p_A(y)$ in the general case.

\section{Algorithm estimating the conditional p-value} \label{section linear simcal algorithm}
Unlike $p_A^0 (y)$, we are able to estimate $p_A (y)$ with a Monte Carlo method. The problem with estimating $p_A^0 (y)$ is that it is not possible to simulate $Y$ following its distribution under $H_0 (A)$, which is not known. To estimate $p_A (y)$, due to the conditioning in its definition, we must instead simulate $Y$ following its distribution under $H_0 (A)$ conditional on $\widehat{\theta_A} (Y)$. The following Monte Carlo algorithm does this. It is parameterized by an integer $N$, the number of simulations, which controls estimation accuracy.

\begin{algorithm}[Estimating $p_A(Y)$ by simulation-calibration]\label{linear simcal algorithm}
The four steps are:
\begin{itemize}
\item[1]{Compute $\widehat{\theta_A} (y)$, the parameter vector of model \ref{modelA} estimated by maximum likelihood.}
\item[2]{Simulate $N$ i.i.d. response vectors $y^{(1)}, .., y^{(N)}$ as follows, for each $l = 1, .. ,N$:
\begin{itemize}
\item[2.1]{Simulate $y^{sim}$ following model \ref{modelA} with an arbitrary parameter vector $\theta_A^{sim} \in \Theta_A^*$:
$$y^{sim} = X_A \beta_A^{sim} + \epsilon \ \ \text{with} \ \epsilon \sim \mathcal{N}(0,{\sigma^{sim}}^2 I_n)$$
}
\item[2.2]{Compute the maximum likelihood estimate $\widehat{\theta_A} (y^{sim})$.}
\item[2.3]{Compute the calibrated version of $y^{sim}$:
$$y^{(l)} = X_A \widehat{\beta_{A}} (y) + \frac{\widehat{\sigma_{A}} (y)}{\widehat{\sigma_{A}} (y^{sim})}  \left(y^{sim} - X_A \widehat{\beta_{A}} (y^{sim}) \right) $$}
\end{itemize}}
\item[3]{For each $l = 1, .., N$, perform the Lasso regression of $y^{(l)}$ on $X$, and compute $\lambda_A (y^{(l)})$.}
\item[4]{Compute the empirical conditional p-value:
$$\widehat{p_A} (y) = \frac{1}{N} \sum_{l=1}^N \mathbf{1} \left\{ \lambda_A (y^{(l)}) \geq  \lambda_A (y) \right\}.$$ }
\end{itemize}
\end{algorithm}

In practice, at step 2.1, we use $\theta_A^{sim} = \widehat{\theta_A} (y)$, but other choices are possible because all values of $\theta_A^{sim}$ produce calibrated responses $y^{(l)}$ that follow the same distribution.

\section{Properties of the algorithm} \label{prop simcal}
\subsection{Proof of consistency and properties for a given response vector}
\subsubsection{Conditional distribution of the response vector}
The algorithm's consistency is based on knowing the conditional distribution of $Y$ under the null hypothesis. It is given by the following lemma:

\begin{lemma} \label{Y uniforme}
Under $H_0 (A)$, the distribution of $Y$ conditional on $\widehat{\theta_A} (Y) = \widehat{\theta_A} (y)$ is the uniform distribution on the set $S \left( \widehat{\theta_A} (y) \right) = \left\{ y' \in \mathbb{R}^n \mid \widehat{\theta_A} (y') =\widehat{\theta_A} (y) \right\}$.
\end{lemma}
The proof is in the supplementary material.
\\
The definition of $p_A (y)$ (\ref{pAy}) is based on an arbitrary random vector $Y$ which follows this conditional distribution. Therefore, $Y$ can be replaced by any random variable following the same distribution, given by lemma \ref{Y uniforme}:
\begin{corollary} \label{pAy toute va}
Any random variable $Y^u$ that follows a uniform distribution over $S \left( \widehat{\theta_A} (y) \right)$ satisfies $\mathrm{P}(\lambda_A (Y^u) \geq  \lambda_A (y)) = p_A (y)$.
\end{corollary}

\subsubsection{Calibration and simulation under the conditional distribution}
Corollary \ref{pAy toute va} implies that to estimate $p_A (y)$ by Monte Carlo, one just has to simulate random response vectors uniformly over $S \left( \widehat{\theta_A} (y) \right)$. Now, we show that step 2. of algorithm \ref{linear simcal algorithm} does this.

We use \textit{calibration functions}, which we define as follows. For any parameter vectors $\theta_{A}^{(1)} = (\beta_{A}^{(1)}, \sigma_A^{(1)}) \in \Theta_A^*$ and $\theta_{A}^{(2)} = (\beta_{A}^{(2)}, \sigma_A^{(2)}) \in \Theta_A$:
\begin{eqnarray}
    \mathrm{cal}_{\theta_{A}^{(1)} \rightarrow \theta_{A}^{(2)}} : \mathbb{R}^n &\rightarrow& \mathbb{R}^n \nonumber \\
     y' &\rightarrow& X_A \beta_{A}^{(2)} + \frac{\sigma_A^{(2)}}{\sigma_A^{(1)}}  \left(y' - X_A \beta_{A}^{(1)} \right). \nonumber
\end{eqnarray}

They map one given value of $\widehat{\theta_A}$ to another, that is:
\begin{lemma} \label{égalité calibration}
For all $\theta_{A}^{(1)} \in \Theta_A^*, \theta_{A}^{(2)} \in \Theta_A$, for all $y^{(1)} \in \mathbb{R}^n$ satisfying $\widehat{\theta_A} (y^{(1)}) = \theta_{A}^{(1)}$, we have:
$$\widehat{\theta_A} \left( \mathrm{cal}_{\theta_{A}^{(1)} \rightarrow \theta_{A}^{(2)}} (y^{(1)}) \right) = \theta_{A}^{(2)}.$$
\end{lemma}
The proof is in the supplementary material.

Calibration functions are a way to "impose" on a response vector $y'$ the condition that $\widehat{\theta_A} (y') = \theta_{A}^{(2)}$ for any target parameter vector $\theta_{A}^{(2)}$. As a composition of a homothety and a translation, they can be seen informally as the simplest change one can make to $y'$ while imposing this condition. In particular, the empirical correlation structure between $y'$ and variables not belonging to $A$ varies little between $y'$ and its calibrated version: it is only modified to the extent that it is carried by the variables in $A$.
\\
Calibration enables simulating the uniform distribution over $S \left( \widehat{\theta_A} (y) \right)$:

\begin{lemma} \label{uniform simulation}
Step 2. of algorithm \ref{linear simcal algorithm} simulates each $y^{(l)}$ following the uniform distribution over $S \left( \widehat{\theta_A} (y) \right)$.
\end{lemma}

\begin{proof}
By replacing $Y$ and $\widehat{\theta_A} (y)$ in lemma \ref{Y uniforme} with $y^{sim}$ and $\widehat{\theta_A} (y^{sim})$, we find that the distribution of $y^{sim}$ conditional on $\widehat{\theta_A} (y^{sim})$ is uniform over $S \left(\widehat{\theta_A} (y^{sim})\right)$. Furthermore, conditional on $\widehat{\theta_A} (y^{sim})$,
$$\mathrm{cal}_{\widehat{\theta_A} (y^{sim}) \rightarrow \widehat{\theta_A} (y)}: y' \rightarrow X_A \widehat{\beta_A} (y) + \frac{\widehat{\sigma_A} (y)}{\widehat{\sigma_A} (y^{sim})} \left(y' - X_A \widehat{\beta_A} (y^{sim}) \right)$$
is the composition of a homothety and a translation. These transformations preserve the uniformity of a probability distribution. Also, according to lemma \ref{égalité calibration}, it maps $S \left( \widehat{\theta_A} (y^{sim}) \right)$ onto $S \left( \widehat{\theta_A} (y) \right)$. Therefore, conditional on $\widehat{\theta_A} (y^{sim})$, $y^{(l)}$ follows the uniform distribution over $S \left( \widehat{\theta_A} (y) \right)$. By integration, without conditioning $y^{(l)}$ follows the uniform distribution over $S \left( \widehat{\theta_A} (y) \right)$.
\end{proof}

\subsubsection{Distribution of the empirical conditional p-value}
The preceding lemmas give us the distribution of $\widehat{p_A} (y)$ produced by the algorithm. It is the binomial distribution:
\begin{lemma} \label{binomial distribution}
$$N \widehat{p_A} (y) \sim \mathrm{Bin} (N, p_A (y)).$$
\end{lemma}
\begin{proof}
For each $l$, $y^{(l)}$ follows the uniform distribution over $S \left( \widehat{\theta_A} (y) \right)$. Therefore, according to corollary \ref{pAy toute va}, for each simulation:

$$P \left( \lambda_A (y^{(l)} ) \geq  \lambda_A (y) \right) = p_A (y).$$

$N \widehat{p_A} (y)$ is thus the sum of $N$ independent binary random variables with mean $p_A (y)$.
\end{proof}

$\widehat{p_A} (y)$ is therefore an unbiased and, due to the law of large numbers, consistent estimator of $p_A (y)$. Its variance is $\mathrm{Var} \left( \widehat{p_A} (y) \right) = p_A (y)(1-p_A (y))/N$, which converges to $0$ as $N$ goes to infinity.

\subsection{Properties for a random response vector} \label{prop simcal random y}

The above applies to a given response vector $y$. In the case of a random response vector $Y$ that follows model \ref{linmodel} (or model \ref{modelA} under the null hypothesis), it is interesting to know the distribution of $\widehat{p_A} (Y)$. That is because we do not observe the theoretical test statistic $p_A (Y)$: we only estimate it. The decision to select or not a variable is based on $\widehat{p_A} (Y)$, and any control of the type I error must be inferred from the distribution of $\widehat{p_A} (Y)$ under the null hypothesis. We determine this distribution exactly under the continuity assumption defined in \ref{conditional pvalue}, and then we infer a stochastic dominance of this distribution in the general case which guarantees the type I error control.
\\
The relationship between $p_A (y)$ and $\widehat{p_A} (y)$ (lemma \ref{binomial distribution}) allows us to translate results about the distribution of $p_A (Y)$ to that of $\widehat{p_A} (Y)$, moving from a continuous to a discrete distribution. We have two results with and without the continuity assumption:

\begin{lemma} \label{discrete uniform}
Under $H_0 (A)$ and under the continuity assumption, the distribution of $\widehat{p_A} (Y)$ is the discrete uniform distribution over $\{0, 1/N, .., 1-1/N, 1\}$.
\end{lemma}
\begin{lemma}\label{discrete stochastic dominance}
Under $H_0 (A)$, the distribution of $\widehat{p_A} (y)$ stochastically dominates the discrete uniform distribution over $\{0, 1/N, .., 1-1/N, 1\}$.
\end{lemma}
This lemmas are analogous to (and derived from) corollary \ref{uniform distribution pAY} and lemmas \ref{inequality pAY} respectively. Their proofs are in the supplementary material; that of lemma \ref{discrete stochastic dominance} uses lemma \ref{discrete uniform} and a criterion for stochastic dominance.
\\
Lemma \ref{discrete stochastic dominance} allows us to control type I error: for all $\alpha \in [0,1]$, producing a $\widehat{p_A} (y)$ and then rejecting $H_0 (A)$ if and only if $\widehat{p_A} (y) \leq \alpha$ ensures that the type I error is less than or equal to $\frac{\lfloor{\alpha N}\rfloor + 1}{N + 1} \leq \alpha + \frac{1 - \alpha}{N + 1}$. Conversely, rejecting $H_0 (A)$ if and only if $\widehat{p_A} (y) \leq \alpha - \frac{1 - \alpha}{N}$ ensures a type I error less than or equal to $\alpha$.

The residual terms $\frac{1 - \alpha}{N + 1}$ and $- \frac{1 - \alpha}{N}$ are very small (smaller than the granularity of estimation of $\widehat{p_A} (y)$, which is $1/N$) and can be controlled by the user since $N$ can be chosen as large as desired (at the cost of computational time). It is also possible to adapt the definition of $\widehat{p_A}$ to formally eliminate these residual terms. With:
$$\widehat{p_A}^{+} (y) = \frac{1}{N+1} \left( 1 + \sum_{l=1}^N \mathbf{1} \left\{ \lambda_A (y^{(l)} ) \geq  \lambda_A (y) \right\} \right),$$
the rejection criterion $\widehat{p_A}^{+} (y) \leq \alpha$ ensures error control at level $\alpha$. However, $\widehat{p_A}^{+} (y)$ is a biased estimator of $p_A (y)$.

\section{Case of generalized linear models} \label{section nonlinear}

\subsection{Problem statement} \label{nonlinear problem statement}
Our algorithm is valid in the the linear model. We propose an adaptation of it to two discrete generalized linear models: the binary model $y|X \sim \mathrm{Bernoulli} \left( f(X \beta) \right)$, and the Poisson model $y|X \sim \mathrm{Poisson} \left( f(X \beta) \right)$. The link function $f$ is the logistic function in the binary model and the exponential in the Poisson model. Unlike the linear one, these two models do not have a standard deviation $\sigma$, therefore we write $\beta_A$ or $\widehat{\beta_A}$ instead of $\theta_A$ or $\widehat{\theta_A}$.

As in the linear case, we assume that the variables $X_j, j \in A$ are active, and we want to know whether any variable outside of $A$ is active. The hypotheses $H_0 (A), H_1 (A)$, and the statistic $\lambda_A (y)$, are defined the same way. Again, ideally we would like to estimate $p_A^0 (y) = \mathrm{P}_{H_0 (A)} \left( \lambda_A (Y) \geq  \lambda_A (y) \right)$, which is not possible because the distribution of $Y$ under $H_0 (A)$ is not known. As in the linear case, we propose an algorithm that computes $\widehat{p_A} (y)$, an estimation of a statistic that approximates $p_A^0 (y)$. This algorithm is similar to the algorithm \ref{linear simcal algorithm}, differing only in its calibration method. However, the statistic estimated by $\widehat{p_A} (y)$ in the generalized linear models is not the same as in the linear case and does not have its theoretical properties.

The solution we chose in the linear case was to estimate the conditional p-value, $p_A (y)$, defined in \ref{pAy}. Unfortunately, in discrete models this quantity is no more interesting because the set $E$ of values that $y$ can take is much more restricted: it is finite in the binary model ($E = \{0,1\}^n$) and countable in the Poisson model ($E = \mathbb{N}^n$) while the set $\Theta_A = \mathbb{R}^{|A|}$ of possible parameters is uncountable. One can generally expect the estimator function $\widehat{\beta_A}$, which sends a finite or countable set to an uncountable set, to be injective: only one vector, $y$, should produce exactly the estimate $\widehat{\beta_A} (y)$. Therefore, the condition $\widehat{\beta_A} (Y) = \widehat{\beta_A} (y)$ implies $Y = y$ hence $\lambda_A (Y) = \lambda_A (y)$ so necessarily $p_A (y) = 1$.

However, if $n$ is large enough, the number of values that $y$ (and thus $\widehat{\beta_A} (y)$) can take is very high ($2^n$ in the binary model). So even if only $y' = y$ satisfies $\widehat{\beta_A} (y') = \widehat{\beta_A} (y)$, a large number of distinct vectors $y'$ can satisfy $\widehat{\beta_A} (y') \approx \widehat{\beta_A} (y)$. Therefore, the exact conditioning can --- informally --- be replaced by a "conditioning by $\widehat{\beta_A} (Y) \approx \widehat{\beta_A} (y)$". More formally, this amounts to using a probability distribution $\mathrm{P}_{\widehat{\beta_A} (y)}$ on the set $E$, producing random vectors $Y$ compatible with $H_{0}(A)$ and which satisfy $\widehat{\beta_A} (Y) \approx \widehat{\beta_A} (y)$. This last condition is quantifiable in terms of mean square error:
$$\mathrm{MSE} \left( \mathrm{P}_{\widehat{\beta_A} (y)} \right) = \mathrm{E}_{Y \sim \mathrm{P}_{\widehat{\beta_A} (y)}} \left[ \left( X_A \widehat{\beta_A} (Y) - X_A \widehat{\beta_A} (y) \right)^2 \right]$$
which must be very close to $0$. We will then estimate:
$$\widetilde{p_A} (y) = \mathrm{P}_{Y \sim \mathrm{P}_{\widehat{\beta_A} (y)}} \left(\lambda_A (Y) \geq \lambda_A (y) \right)$$
which is the approximate equivalent in discrete models of the test statistic $p_A (y)$ of the linear model.

In the linear model, we simulated vectors $y^{(l)}$ following the distribution produced by conditioning on $\lambda_A (Y) = \lambda_A (y)$ (a distribution we explicitly determined) and then used the population of $y^{(l)}$ to estimate $p_A(y)$ via Monte Carlo. We simulated vectors following this distribution by using a calibration function on vectors that were first simulated following the linear model under the null hypothesis. Similarly, in discrete models, we propose to simulate $y^{(l)}$ following a distribution $\mathrm{P}_{\widehat{\beta_A} (y)}$ that meets the conditions above but is not explicitly determined, and to use these $y^{(l)}$ to estimate $\widetilde{p_A} (y)$ via Monte Carlo. The simulation of vectors following $\mathrm{P}_{\widehat{\beta_A} (y)}$ is also carried out by applying a calibration procedure to vectors simulated according to the generalized linear model under the null hypothesis. This calibration is, however, more complex than that of the linear model.

\subsection{Calibration in nonlinear models} \label{nonlinear calibration}

Given two parameter vectors $\beta_{A}^{(1)}, \beta_{A}^{(2)} \in \mathbb{R}^n$, we want a calibration procedure from $\beta_{A}^{(1)}$ to $\beta_{A}^{(2)}$ such that:
\begin{conditions}[Desirable properties of calibration] \label{calibration conditions}
For every $y^{(1)} \in E$, denoting $y^{(2)}$ as its calibrated version:
\begin{itemize}
\item[\ref{calibration conditions}.1.]{$y^{(2)}$ should be "close" to $y^{(1)}$ to preserve as much as possible its correlation structure with variables outside of $A$.}
\item[\ref{calibration conditions}.2.]{If $\widehat{\beta_A} (y^{(1)}) = \beta_{A}^{(1)}$ then
$\widehat{\beta_A} ( y^{(2)} ) \approx \beta_{A}^{(2)}$.}
\end{itemize}
\end{conditions}
Condition \ref{calibration conditions}.2 is the informal adaptation of the formal lemma \ref{égalité calibration} satisfied by the linear calibration. The equality between parameters, which is a property too strong in discrete models, is replaced by an approximation.

In linear models, calibration is a function $\mathrm{cal}_{\beta_{A}^{(1)} \rightarrow \beta_{A}^{(2)}} : E \rightarrow E$ involving multiplication by and addition of non-integer constants. In discrete models, this is not possible because the calibration must produce a vector of integer numbers. Our solution is to replace deterministic non-integers with random integers having the same expected value. Therefore, calibration in nonlinear models is a random procedure described by the conditional distribution of the calibrated vector given the initial vector: $Y^{(2)} | Y^{(1)} \sim \mathrm{P}_{\beta_{A}^{(1)} \rightarrow \beta_{A}^{(2)}}$.

\subsubsection{One-step calibration algorithm in nonlinear models}
Below we present a one-step calibration procedure, which follows a conditional distribution denoted $\mathrm{P}^{(1)}_{\beta_{A}^{(1)} \rightarrow \beta_{A}^{(2)}}$. It is a step of the complete procedure that simulates vectors following $\mathrm{P}_{\beta_{A}^{(1)} \rightarrow \beta_{A}^{(2)}}$. In both generalized linear models, it is based on the prediction vectors produced by the initial and target parameters: $e^{(1)} = f \left( X \beta_{A}^{(1)} \right)$ and $e^{(2)} = f \left( X \beta_{A}^{(2)} \right)$.

\begin{algorithm}[Non-linear calibration: one-step algorithm]\label{calibration iteration} Given $Y^{(1)}$, simulating $Y^{(2)} | Y^{(1)} \sim \mathrm{P}^{(1)}_{\beta_{A}^{(1)} \rightarrow \beta_{A}^{(2)}}$ means simulating the vector $Y^{(2)}$ independently for each individual $i$: \begin{itemize}
\item{In the binary model: \begin{eqnarray}
Y^{(2)}_i | Y^{(1)}_i &\sim& \mathrm{Bernoulli} (Z_i) \ \ \text{where:} \nonumber \\
\text{if} \ e^{(2)}_i \leq e^{(1)}_i, \ Z_i &=& \frac{e^{(2)}_i}{e^{(1)}_i} Y^{(1)}_i \nonumber \\
\text{if} \ e^{(2)}_i \geq e^{(1)}_i, \ Z_i &=& 1 - \frac{1 - e^{(2)}_i}{1 - e^{(1)}_i} (1 - Y^{(1)}_i) \nonumber
\end{eqnarray}}
\item{In the Poisson model:
\begin{eqnarray}
Y^{(2)}_i &=& \left\lfloor Z_i \right\rfloor + R_i \ \ \text{where:} \nonumber \\
Z_i &=& \frac{e^{(2)}_i}{e^{(1)}_i} Y^{(1)}_i \nonumber \\
R_i | Y^{(1)}_i &\sim& \mathrm{Bernoulli} \left( Z_i - \left\lfloor Z_i \right\rfloor \right). \nonumber
\end{eqnarray}}
\end{itemize}
\end{algorithm}

\subsubsection{Properties of the one-step calibration algorithm}
The two non-exclusive scenarios of the binary case are consistent because if $e^{(1)}_i = e^{(2)}_i$ then $Y^{(2)}_i = Y^{(1)}_i$ regardless of the formula used.  In both models, $\mathrm{E} \left[Y^{(2)} | Y^{(1)} \right] = Z$ and if $\beta_A^{(1)} = \beta_A^{(2)}$ (the case where no calibration is necessary) then $e^{(1)} = e^{(2)}$ and $Y^{(1)} = Z = Y^{(2)}$.

If an individual $i$ satisfies $e^{(1)}_i \approx e^{(2)}_i$, then the probability that $Y^{(2)}_i = Y^{(1)}_i$ is high. If $\beta_A^{(1)} \approx \beta_A^{(2)}$, this is usually true, thus, as desired (condition \ref{calibration conditions}.1.), calibration does not modify much of $Y^{(1)}$.

Furthermore, the calibration "transforms a vector following the parameters $\beta_A^{(1)}$ into a vector following the parameters $\beta_A^{(2)}$" in the sense that, with $Y^{(1)}$ being random:
\begin{eqnarray} \label{parameter transfer}
\mathrm{E} \left[Y^{(1)} \right] = e^{(1)} \implies \mathrm{E} \left[Y^{(2)} \right] = \mathrm{E} \left[Z \right] = e^{(2)}.
\end{eqnarray}

However, this property does not imply condition \ref{calibration conditions}.2., which is about the estimates $\widehat{\beta_A} (Y^{(1)})$ and $\widehat{\beta_A} (Y^{(2)})$, or, equivalently, the prediction vectors they produce: $\hat{e}^{(m)} = f \left( \widehat{\beta_A} (Y^{(m)}) \right)$, $m = 1$ or $2$. The condition can be verified in the simple case where the model is reduced to an intercept, that is, $A = \emptyset$. The prediction vectors $e^{(1)}$ and $e^{(2)}$ are then constants.

\begin{lemma} \label{single parameter calibration}
In single-parameter models, if $\widehat{\beta_A} (Y^{(1)}) = \beta_{A}^{(1)}$ that is, $\hat{e}^{(1)} = e^{(1)}$, then:
\begin{itemize}
\item $\mathrm{E} \left[ \hat{e}^{(2)} | Y^{(1)} \right] = e^{(2)}$
\item $\mathrm{Var} \left( \hat{e}^{(2)} | Y^{(1)} \right) \leq \frac{1}{n} \left| e^{(2)} - e^{(1)} \right|$ with equality in the binary model.
\end{itemize}
\end{lemma}

We prove it in the supplementary material. In single-parameter models, $\mathrm{P}^{(1)}_{\widehat{\beta_{A}} (Y^{(1)}) \rightarrow \beta_{A}^{(2)}}$ thus produces a distribution of $Y^{(2)}$ where $\hat{e}^{(2)}$ is distributed around $e^{(2)}$ with low variance. Therefore, by applying $f^{-1}$, $\widehat{\beta_{A}} (Y^{(2)})$ is distributed around the target value $\beta^{(2)}$ with low variance. This result, combined with the parameter transfer property from $\beta^{(1)}$ to $\beta^{(2)}$ \eqref{parameter transfer} which is valid in multi-parameter models, suggests that the distribution of $\widehat{\beta_{A}} (Y^{(2)})$ under $\mathrm{P}^{(1)}_{\widehat{\beta_{A}} (Y^{(1)}) \rightarrow \beta_{A}^{(2)}}$ has a similar property in multi-parameter models.

\subsubsection{Iterative calibration algorithm}
The  variance inequality of lemma \ref{single parameter calibration} indicates that, at least in single-parameter models, $\hat{e}^{(2)}$ (or $\widehat{\beta_{A}} (Y^{(2)})$) approaches its target with greater precision when the initial vector ($\hat{e}^{(1)}$ or $\widehat{\beta_{A}} (Y^{(1)})$) is already closer to it. It is therefore useful to reapply the calibration procedure following a distribution $\mathrm{P}^{(1)}_{\cdot \ \rightarrow \beta_{A}^{(2)}}$ to an already calibrated vector. We thus define of $\mathrm{P}_{\beta_{A}^{(1)} \rightarrow \beta_{A}^{(cal)}}$, the conditional distribution describing the complete calibration procedure, via the following simulation algorithm.

\begin{algorithm}[Non-linear calibration: iterative algorithm]\label{nonlinear iterative calibration} Let $\beta_{A}^{(1)}, \beta_{A}^{(cal)} \in \Theta$ be parameter vectors and $y^{(1)}$ a response vector. Simulating $Y^{(cal)} | y^{(1)} \sim \mathrm{P}_{\beta_{A}^{(1)} \rightarrow \beta_{A}^{(cal)}}$ means, for a certain number of iterations $k = 1, 2, ..$:
\begin{itemize}
\item[1.] Simulate a vector $y^{(k+1)}$ using $y^{(k)}$ following the distribution $Y^{(k+1)} | Y^{(k)} \sim \mathrm{P}^{(1)}_{\beta_{A}^{(k)} \rightarrow \beta_{A}^{(cal)}}$, with algorithm \ref{calibration iteration}.
\item[2.] Compute the maximum likelihood estimator, and define $\beta_{A}^{(k+1)} = \widehat{\beta_A} (y^{(k+1)})$.
\item[3.] Compute the empirical mean squared error
$\mathrm{MSE}(k+1) = {|| X_A \beta_{A}^{(k+1)} - X_A \beta_{A}^{(cal)} ||}_2^2$. If $\mathrm{MSE}(k+1) > \mathrm{MSE}(k)$, reject the $y^{(k+1)}$ that has just been simulated and instead define $y^{(k+1)} = y^{(k)}, \ \beta^{(k+1)} = \beta^{(k)}$. If this occurs 3 times consecutively, end the iterations.
\end{itemize}
The value of $y^{(k)}$ at the time of stopping the iterations is then taken as $Y^{(cal)}$.
\end{algorithm}

\subsection{Simulation-calibration testing algorithm in generalized linear models}

With the calibration algorithms \ref{calibration iteration} and \ref{nonlinear iterative calibration}, we can propose the algorithm for generating the empirical conditional p-value $\widehat{p_A} (y)$ in binary and Poisson generalized linear models. It is similar in form to algorithm \ref{linear simcal algorithm} for linear models, with two main differences:
\begin{itemize}
\item theoretically, the quantity estimated by $\widehat{p_A} (y)$ is not the same conditional p-value, $\widetilde{p_A}(Y)$ instead of $p_A(Y)$ (see \ref{nonlinear problem statement});
\item algorithmically, the calibration step is more complex, using the sub-algorithms \ref{calibration iteration} and \ref{nonlinear iterative calibration}.
\end{itemize}

\begin{algorithm}[Estimation of $\widetilde{p_A}(Y)$ by simulation-calibration]\label{nonlinear simcal algorithm}
The four steps are as follows:
\begin{itemize}
\item[1]{Compute $\widehat{\beta_A} (y)$, the parameter vector of the model restricted to $A$ estimated by maximum likelihood.}
\item[2]{Simulate $N$ independent response vectors $y^{(1)}, .., y^{(N)}$ each following a distribution denoted $P_{\widehat{\beta_A} (y)}$. This means, for each $l = 1, .. ,N$:
\begin{itemize}
\item[2.1]{Simulate $y^{sim}$ according to the generalized linear model with the parameter vector $\widehat{\beta_A} (y)$:
\begin{eqnarray}
y^{sim} &\sim& \mathrm{Bernoulli} \left(\frac{1}{1 + \exp(-X_A \widehat{\beta_A} (y))} \right) \ \text{or} \nonumber \\
y^{sim} &\sim& \mathrm{Poisson} \left(\exp( X_A \widehat{\beta_A} (y) ) \right). \nonumber
\end{eqnarray}}
\item[2.2]{Compute the maximum likelihood estimate $\widehat{\beta_A} (y^{sim})$.}
\item[2.3]{Generate a calibrated version of $y^{sim}$ towards $\widehat{\beta_A} (y)$ following algorithm \ref{nonlinear iterative calibration}:
$$y^{(l)} | y^{sim} \sim \mathrm{P}_{\widehat{\beta_A} (y^{sim}) \rightarrow \widehat{\beta_A} (y)}.$$}
\end{itemize}}
\item[3]{For each $l = 1, .., N$, perform the Lasso regression of $y^{(l)}$ on $X$, and compute $\lambda_A (y^{(l)})$.}
\item[4]{Compute the empirical conditional p-value:
$$\widehat{p_A} (y) = \frac{1}{N} \sum_{l=1}^N \mathbf{1} \left\{ \lambda_A (y^{(l)} ) \geq  \lambda_A (y) \right\}.$$ }
\end{itemize}
\end{algorithm}

$\widehat{p_A} (y)$ is intended to be interpreted the same way in linear and generalized linear models. However, the properties relating to its distribution under the null hypothesis (lemmas \ref{discrete uniform} and \ref{discrete stochastic dominance}) are only demonstrated in the linear case. The construction of $\widehat{p_A} (y)$ in generalized linear models is designed so that these properties are approximately valid there.

\section{Variable selection procedure} \label{section proc simcal}

The simulation-calibration test measures the significance of a covariate selected by the Lasso. When no active variables are known a priori and a complete model needs to be selected, it is necessary to use the test iteratively.

\subsection{Notations and algorithm}
For every $k \in \{1, .., p\}$, let $j_k$ be the index of the $k$-th variable selected by the Lasso and let $A_k = \{ j_1, .., j_k \}$ be the set of the first $k$ variables selected, with $A_0 = \emptyset$. Each $j_k$ is the first variable selected outside the set $A_{k-1}$, that is, following to the notation introduced in \ref{problem statement simcal}, $j_k = j_{A_{k-1}}$. This is a slight variant of the Lasso path where variables do not "reenter": if there is a $\lambda$ such that $\widehat{\beta_j}^{Lasso} (\lambda) \neq 0$ (which translates into a $k$ such that $j \in A_k$), then by definition $j$ belongs to all $A_{k'}, k' > k$ even when there are $\lambda' < \lambda$ such that $\widehat{\beta_j}^{Lasso} (\lambda') = 0$.

The simulation-calibration test of $H_0 \left(A_{k-1} \right)$ thus measures the significance of the variable $j_k$. Furthermore, since variables do not "reenter" the sequence of sets, the null hypotheses $H_0 \left(A_{k} \right)$ become progressively weaker: if $k < k'$, then $A_k \subset A_{k'}$ and since the hypotheses concern the complements of these sets (on which they assert the nullity of $\beta$), $H_0 \left(A_{k} \right)$ implies $H_0 \left(A_{k'} \right)$.

\begin{algorithm}[Simple variable selection procedure by simulation-calibration]\label{algorithme simcal itere}
For $k = 1, 2, ..$
\begin{itemize}
\item[1.]{Compute $p_k = \widehat{p_{A_{k-1}}} (y)$ by simulation-calibration (algorithm \ref{linear simcal algorithm} in the linear case, or \ref{nonlinear simcal algorithm} in the nonlinear cases).}
\item[2.]{\begin{itemize}
    \item If $p_k \leq \alpha$, then continue the algorithm;
    \item otherwise, select $A_{k-1}$ and halt the algorithm.
\end{itemize}}
\end{itemize}
\end{algorithm}

\subsection{Choice of the halting criterion}\label{ForwardStop}

In this procedure, we carry out several hypothesis tests. However, since the procedure halts at the first test where the null hypothesis is not rejected, it cannot be considered a multiple testing procedure per se, in which a potentially high number of tests satisfy their null hypothesis. Therefore, it is not advisable to make the rejection threshold more stringent as done by procedures like Bonferroni and Benjamini and Hochberg \citep{benjamini_controlling_1995}, which are used to control the risks of false positives deriving from the large number of p-values generated under the null hypothesis.

However, it is possible to adapt the halting criterion of the procedure to the fact that it evaluates a sequence of ordered tests. In the general problem of a sequence of ordered tests whose p-values $p_1, .., p_m$ are measured, \cite{gsell_sequential_2016} proposed the \textit{ForwardStop} criterion. This rejects the first $\hat{k}_F$ tests where:
\begin{eqnarray}
\hat{k}_F &=& \max \left\{ k \in \{1, .., m \} : p_k^{FS} \leq \alpha \right\} \nonumber \\
p_k^{FS} &=& -\frac{1}{k} \sum_{i = 1}^{k} \log(1 - p_i). \nonumber
\end{eqnarray}

This criterion has the advantage that, in the vast majority of cases, one can be compute it knowing only the first values of the sequence $(p_k)$, because the sequence $(p_k^{FS})$ is calculated from $p_i, i \leq k$ and it generally increases with $k$. In practice, when we apply ForwardStop, we select the first $\hat{k}_F'$ variables, where:
$${\hat{k}_F}' = \min \left\{ k \in \{1, .., m \} :  p_k^{FS} > \alpha \right\} - 1$$
which amounts to replacing the condition to continue the algorithm at step 2 with $p_k^{FS} \leq \alpha$.

In contrast to ForwardStop, the simple criterion $p_k \leq \alpha$ is called \textit{thresholding}. Thresholding controls both the FWER and the FDR at level $\alpha$ \citep{marcus_closed_1976}, while ForwardStop, generally less conservative, controls the FDR at level $\alpha$ \citep{gsell_sequential_2016}.

These control results apply to the type I error, that is, incorrect rejection of the null hypothesis when it is true. However, in the sequential procedure, we often test $H_0 (A)$ hypotheses that are not true: notably at the first step of the procedure, where $A = \emptyset$ and it is sufficient that there is an active variable for $H_0 (A)$ to be incorrect. It is possible that there are active variables outside of $A$ but that they are selected by the Lasso only at low values of $\lambda$ and that the first variable selected outside of $A$ is an inactive variable. In this case, retaining this variable constitutes a false positive from the perspective of variable selection, but not a type I error since $H_0 (A)$ is not true. Section \ref{extended theorem} presents a result controlling this occurrence.

\section{Extended theorem: control of the selection error} \label{extended theorem}
We divide the set of covariates into three disjoint subsets: $\{1, .. ,p \} = A \cup B \cup C$ with $\forall j \in C, \beta_j = 0$.

We perform the simulation-calibration test of the hypothesis $H_0 (A)$ with the goal of variable selection, that is, if the test rejects $H_0 (A)$, we select the (almost certainly unique) variable $j_A (Y) \in B\cup C$ such that $\widehat{\beta_{j_A (Y)}}^{Lasso} (\lambda) \neq 0$ for $\lambda$ near $\lambda_A$. This selection is a true positive if $j_A (Y) \in B$ and a false positive if $j_A (Y) \in C$.
Since $H_0 (A)$ is not true (because there may exist $j \in B$ such that $\beta_j \neq 0$), the results from the previous sections do not apply. However, we have the following result:
\begin{theorem}
Assume that the active variables are orthogonal to the inactive variables, that is:
$$X_C^T X_{A\cup B}=0.$$
Then, the simulation-calibration test of $H_0 (A)$ at level $\alpha$ has a probability less than $\alpha + \frac{1 - \alpha}{N+1}$ of selecting a false positive.
\end{theorem}
We prove it in the supplementary material.
\\
Like under the null hypothesis (section \ref{prop simcal random y}), this result can also be understood as controlling at the level $\alpha$ the risk of selecting a false positive if the selection criterion is $\widehat{p_A} (y) \leq \alpha - \frac{1 - \alpha}{N}$.
\\
The theorem renders applicable the properties of control of the FWER and FDR by the variable selection procedure with thresholding and ForwardStop respectively, provided that the active variables are orthogonal to the inactive variables.

\section{Simulation studies} \label{simulation studies}
To assess the performance of the test by simulation-calibration, we conducted two simulation studies.

First, we measured the distribution of the p-value generated under $H_0 (A)$ in a large number of different scenarios. By construction, this distribution is supposed to be the uniform on $[0,1]$. More precisely, we proved in Lemma \ref{discrete uniform} that in the linear model, under the continuity assumption, $\widehat{p_A} (Y)$ follows the uniform distribution on a discrete analogue of $[0,1]$, and in Lemma \ref{discrete stochastic dominance} that by relaxing the continuity assumption, in the linear model $\widehat{p_A} (Y)$ stochastically dominates the uniform distribution. However, we do not have equivalent theorems in generalized linear models, although the properties of nonlinear calibration demonstrated in section \ref{nonlinear calibration} suggest that a distribution close to the uniform is also to be expected in these models.

We also measured the performance of the variable selection procedures —-- with both the thresholding and ForwardStop halting criteria —-- using three usual metrics of variable selection: family-wise error rate (FWER), false discovery rate (FDR), and sensitivity. Due to the properties of thresholding and ForwardStop, we expect control of the FWER with the former and of the FDR with the latter, at least under the conditions where the extended theorem applies: in the linear model and with no correlation between covariates. Additionally, we compared these performances with those of an equivalent procedure based on the CovTest by \cite{lockhart_significance_2014}.

\subsection{Simulation plan}
The following simulation plan is shared by the study of the p-value's distribution under the null and by that of the variable selection procedure. We simulate $n_{\mathrm{sim}} = 500$ data sets for each of the 252 parameter sets (or scenarios). In all cases, the number of observations is $n = 1000$ and the number of covariates is $p = 500$. The parameters which vary across scenarios are:
\begin{itemize}
\item{The type of model: linear, binary with dense data, binary with sparse data, or Poisson. In binary models, dense data means that $\mathrm{E}[Y|X = 0] = 0.5$ and sparse data that $\mathrm{E}[Y|X = 0] = 0.1$, with the value of $\beta_0$ distinguishing these two cases.}
\item{The correlation matrix used to simulate the regressors: a Toeplitz matrix of coefficients $\rho_{(i,j)}=\rho^{\lvert i - j \rvert}$, with $\rho=0$, $\rho=0.9$, or $\rho=0.99$. These high values are used by \cite{sabourin_permutation_2015} in their simulation study.}
\item{The number of active variables: $0, 1, 2, 5$, or $10$. They are drawn uniformly among the 500 covariates.}
\item{For scenarios with at least one active variable, the signal-to-noise ratio (SNR). This quantity is inspired by the SNR used by Sabourin et al. Unlike them, we use an empirical version which is defined for each generalised linear model. This is the ratio between the empirical variance of the signals ($E_{\beta}[Y_i \mid X], i = 1, ... , n$) and the empirical mean of the variances of the noises ($\mathrm{Var}_{\beta}(Y_i \mid X), i = 1, ... , n$): $E_{\beta}[Y_i \mid X], i = 1, ... , n$):
$$\
\mathrm{SNR}(X, \beta) = \frac{\frac{1}{n-1} \Sigma_{i=1}^n \left(E_{\beta}[Y_i \mid X] - \frac{1}{n} \Sigma_{i=1}^n E_{\beta}[Y_i \mid X] \right)^2 }{\frac{1}{n} \Sigma_{i=1}^n \mathrm{Var}_{\beta}(Y_i \mid X)}.$$ 
A higher SNR means that the impact of each active variable on $Y$ is more easily observable, which should make variable selection more effective. Scenarios with $0$ active variables necessarily have a zero signal-to-noise ratio. We set $\mathrm{SNR}(X, \beta) = 1, 0.3, 0.1, 0.03, \ \text{or} \ 0.01$.}
\end{itemize}

\subsection{p-value under the null hypothesis} \label{etude simu simcal H0}

To verify that the test follows its expected behavior under the null hypothesis, we assume that the set $A$ of $0$ to $10$ active regressors is known and we test $H_0 (A)$. For each scenario, in each of the $n_{\mathrm{sim}} = 500$ data sets $s$ characterized by $A_s$, $X_s$, and $Y_s$, we produce a $\hat{p}_s = \widehat{p_{A_s}} (Y_s, X_s)$ using the simulation-calibration algorithm based on $N = 100$ simulations of calibrated response vectors.
\\
In one of the scenarios (linear model, $\rho = 0.99$, 1 active regressor, $\mathrm{SNR} = 1$), as an example, we also produced in each of the 500 data sets a naive estimate of the unconditional p-value $p_A^0 (y)$ (see section \ref{conditional pvalue}). It is obtained by Monte Carlo without the calibration step, i.e. by applying the algorithm \ref{linear simcal algorithm} without the calibration step (2.3.), taking $y^{(l)} = y^{sim}$ with $\theta_A^{sim} = \widehat{\theta_A} (y)$. This intends to illustrate the impact of calibration on the distribution of the p-values.
\\
In each scenario, we observe whether the population of ${\left( \hat{p}_s \right)}_{1 \leq s \leq n_{\mathrm{sim}}}$ is distinguishable from a sample drawn the uniform distribution on $[0,1]$. Graphically, agreement or disagreement with the uniform distribution is observed on quantile-quantile (q-q) plots where for each $s = 1, .., n_{\mathrm{sim}}$, the $s$-th smallest value $\hat{p}_{(s)}$ is represented at the coordinates $\left( \frac{s}{n_{\mathrm{sim}}}, \hat{p}_{(s)} \right)$.
\\
The diagrams in the supplementary material illustrate the necessity of the calibration step for the empirical p-values produced to be valid. In this example (linear model, $\rho = 0.99$, $1$ active variable, $\mathrm{SNR} = 0.1$), p-values produced by simulation-calibration are compatible with the uniform distribution. In contrast, p-values produced without calibration deviate significantly, the smallest p-value among 500 being $0.13$. At usual test levels, false positives are thus practically impossible instead of being possible with a controlled probability, suggesting a very low power of the test. At higher test levels, type I error is no longer controlled.
\\
To systematically evaluate the adequacy of the $\hat{p}_s$ to the uniform distribution across all scenarios, we used the Kolmogorov-Smirnov (K-S) test on the population of ${\left( \hat{p}_s \right)}_{1 \leq s \leq n_{\mathrm{sim}}}$ in each scenario. We conducted the test in its two-sided version, where the null hypothesis is that the population considered is a sample from the target distribution (here the uniform distribution on $[0,1]$), and in one of its one-sided versions having a less strict null hypothesis: that the population is sampled from a distribution which has stochastically dominates the target distribution. We are interested in this one-sided test because, as seen in the conclusion of section \ref{prop simcal random y}, basing a test on an empirical p-value whose distribution dominates the uniform distribution on $[0,1]$ allows control of its type I error.
\\
Results vary across the model type: linear, dense binary, sparse binary, or Poisson. For each of the four types and each variant of the Kolmogorov-Smirnov test, we used the Bonferroni correction on all the p-values of the 63 K-S tests applied to the scenarios of that model type.

\begin{figure}
    \centering
    \makebox[\textwidth][c]{\includegraphics[width=\textwidth]{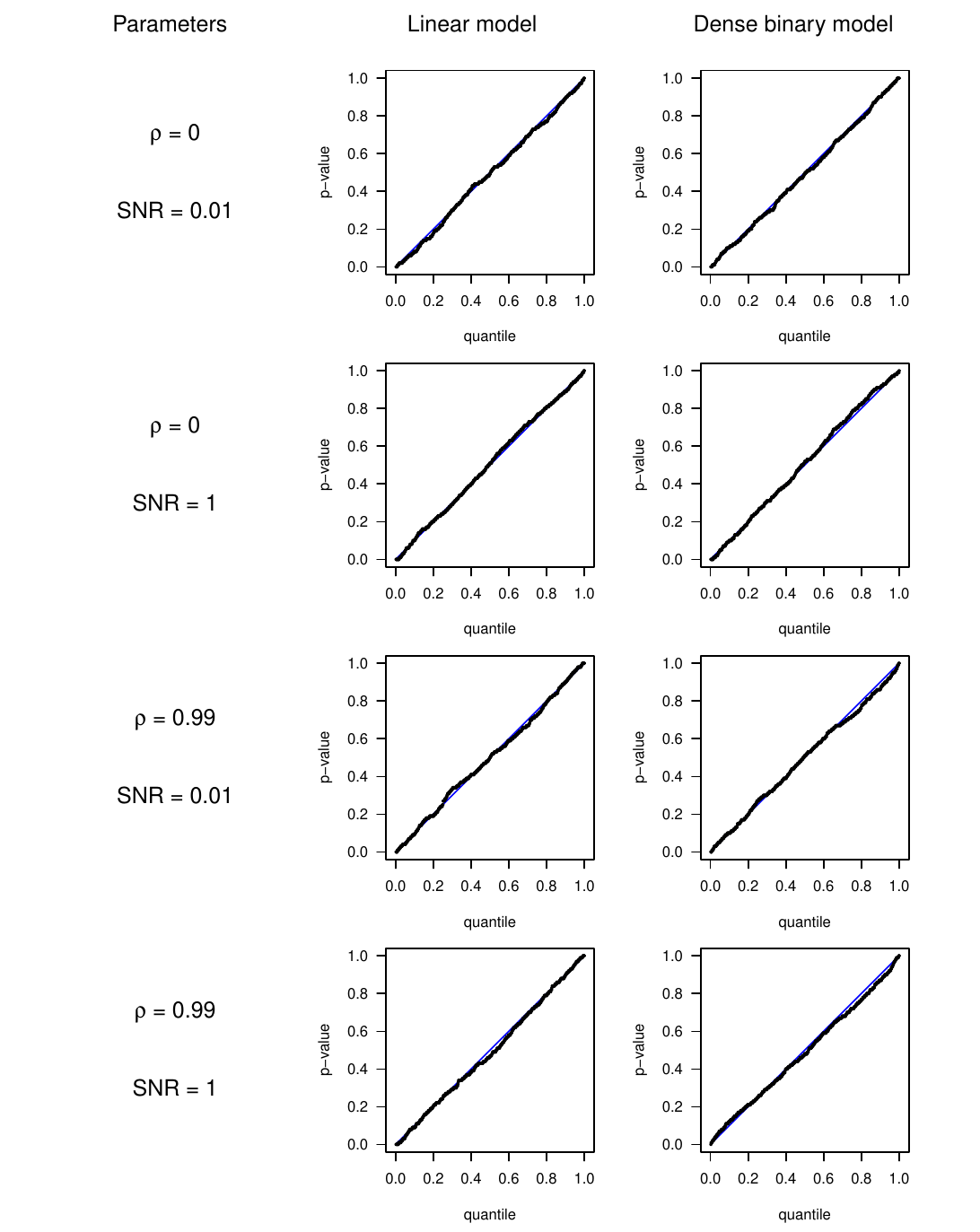}}
    \caption{Quantile-quantile diagrams of 500 empirical p-values generated by simulation-calibration in 8 scenarios of linear model or dense binary model with 10 active variables.}
    \label{fig:simcal-result-H0-1}
\end{figure}

\begin{figure}
    \centering
    \makebox[\textwidth][c]{\includegraphics[width=\textwidth]{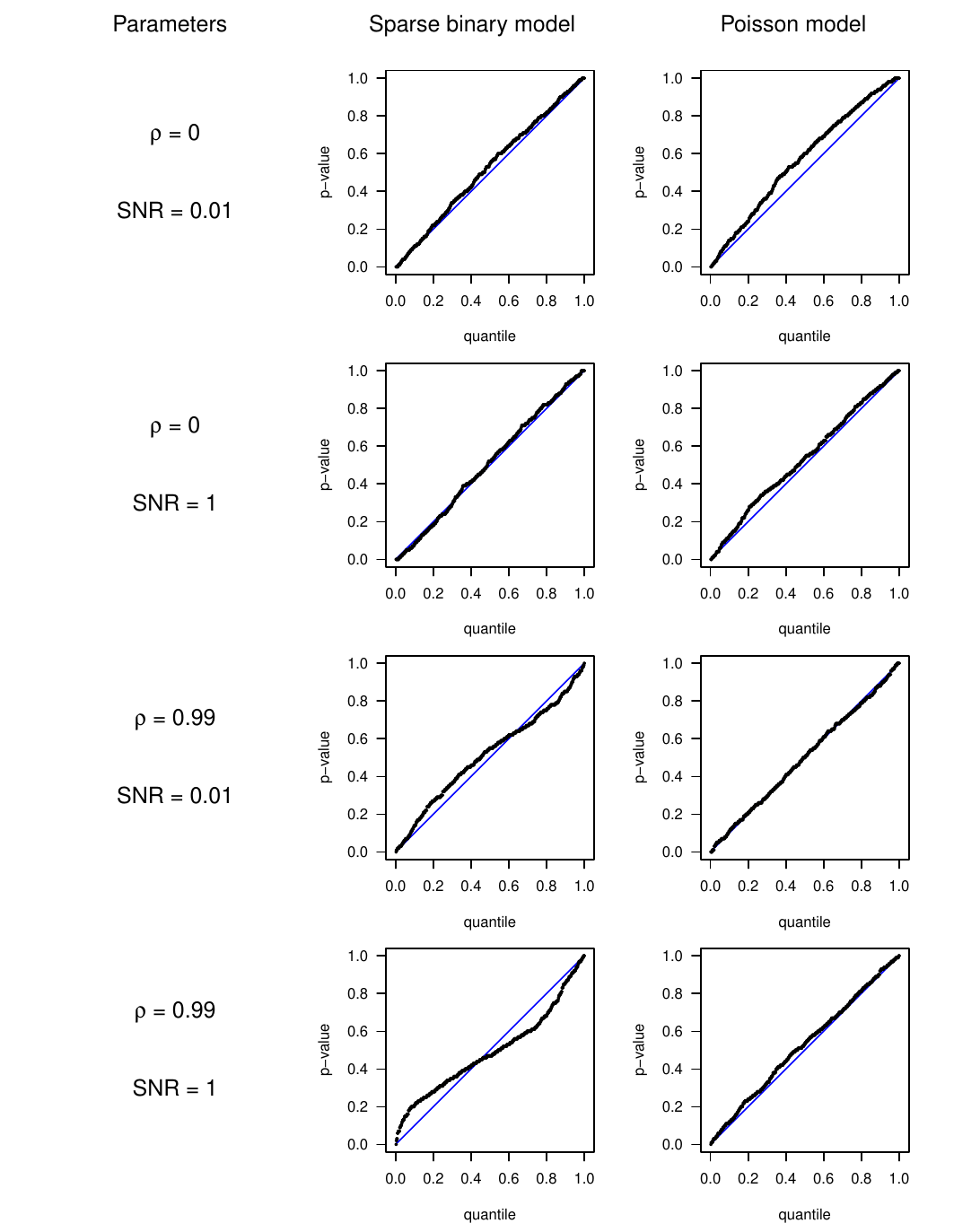}}
    \caption{Quantile-quantile diagrams of 500 empirical p-values generated by simulation-calibration in 8 scenarios of sparse binary model or Poisson model with 10 active variables.}
    \label{fig:simcal-result-H0-2}
\end{figure}

In the linear and dense binary models, the bilateral K-S tests do not reject the hypothesis that each population of p-values follows a uniform distribution over $[0,1]$. The smallest p-value is $0.0112$ among the linear scenarios and $0.0546$ among the balanced binary scenarios, which, combined with the Bonferroni correction for 63 tests, does not lead to rejection at usual test levels ($63 \times 0.0112 = 0.706$). Figure \ref{fig:simcal-result-H0-1} illustrates this general adequacy with the uniform distribution across eight example scenarios.
\\
In the sparse binary model, there are scenarios where the distribution of generated p-values significantly deviates from the uniform distribution and does not dominate it. At the $0.1$ level and with the Bonferroni correction, the null hypothesis of the K-S test (unilateral as well as bilateral) is rejected in two out of 63 scenarios: those with maximum correlation among covariates ($\rho = 0.99$) and the number of active regressors ($|A| = 10$), and where the signal-to-noise ratio is high ($\mathrm{SNR} = 0.3$ or $\mathrm{SNR} = 1$). The p-values of the bilateral K-S test in these scenarios are respectively $p = 2.86 \times 10^{-6} = 1.81 \times 10^{-4} / 63$ and $p = 2.55 \times 10^{-7} = 1.61 \times 10^{-5} / 63$. Despite this non-adequacy, the q-q plots (Figure \ref{fig:simcal-result-H0-2}, which includes the $\mathrm{SNR} = 1$ scenario where the deviation from the uniform distribution is most pronounced) show that $P_{H_0 (A)}(\hat{p}_A \leq \alpha) \leq \alpha$ at low values of $\alpha$, the opposite occurring only when $\alpha$ exceeds about $0.4$. Then, in practice, we have a more conservative control of the type I error than the nominal level at usual test levels. In the most divergent scenario, the diagram even shows a type I error significantly lower than its expected level, signaling a loss of test power. However, this effect is considerably less pronounced than in the uncalibrated simulation example.
\\
In the Poisson model, the unilateral K-S tests do not reject the hypothesis that each population of p-values stochastically dominates the uniform distribution over $[0,1]$ (the minimum p-value across unilateral tests is $0.0214 = 1.35/63$). Thus, FWER control is at least as conservative as the nominal level at all test levels. However, at the $0.1$ level and with the Bonferroni correction, the bilateral K-S tests reject the hypothesis of a uniform distribution of p-values in 5 out of 63 scenarios. These scenarios present a different profile from those where deviations from the uniform distribution were observed in the sparse binary case: 4 out of 5 are scenarios with zero correlation, all have an SNR of $0.1$ or less, and the minimum p-value ($p = 1.79 \times 10^{-6} = 1.13 \times 10^{-4} / 63$) is reached at zero correlation, minimum signal-to-noise ratio ($\mathrm{SNR} = 0.01$), and $|A| = 10$ known active regressors. Moreover, as seen in Figure \ref{fig:simcal-result-H0-2}, where the scenario with the minimum K-S test p-value appears, even then the deviation from the uniform is small, especially at the small quantiles. The actual level of first type error is therefore very close to its nominal level at its usual values.

\subsection{Variable selection procedure}

In this simulation study, unlike in section \ref{etude simu simcal H0}, we did not assume the set of active regressors to be known. This set is estimated by the sequential procedure described in algorithm \ref{algorithme simcal itere}. The number of simulations in the calculation of empirical p-values is higher, $N = 500$, which allows for a more precise estimation of p-values.
\\
At first, we applied the algorithm with the thresholding halting criterion at $\alpha = 0.95$, a high level that produces on each dataset a relatively complete sequence of variables that may be selected, along with their associated p-values. In a second step, we used the variable selection procedure on these sequences of p-values for each $\alpha$ in a relatively dense grid of values (from $0.01$ to $0.5$, with a step of $0.01$), with both thresholding and ForwardStop. In this second step, we observe the procedure's performance on a large number of values of $\alpha$ in a computationally affordable manner since it does not require recomputing the empirical p-values through simulation-calibration, as this was done in the first step.

\subsubsection{Control of FWER and FDR}

Figure \ref{fig:procedure-simcal-FWER-1} and \ref{fig:procedure-simcal-FDR-1} show how the FWER and FDR of the thresholding and ForwardStop procedures vary with $\alpha$ in the same 8 example scenarios as in Figure \ref{fig:simcal-result-H0-1}: 10 active variables, either zero or maximal correlation between covariates, minimal or maximal SNR, linear or dense binary model. In the the supplementary material, Figures 1 and 2 of  show the variation of the FWER and FDR of both procedures with $\alpha$ in the 8 scenarios of Figure \ref{fig:simcal-result-H0-1} (sparse binary and Poisson models with the same other parameters) while Figures 3 and 4 display the variations of their sensitivity in those 16 scenarios.

\begin{figure}
    \centering
    \makebox[\textwidth][c]{\includegraphics[width=\textwidth]{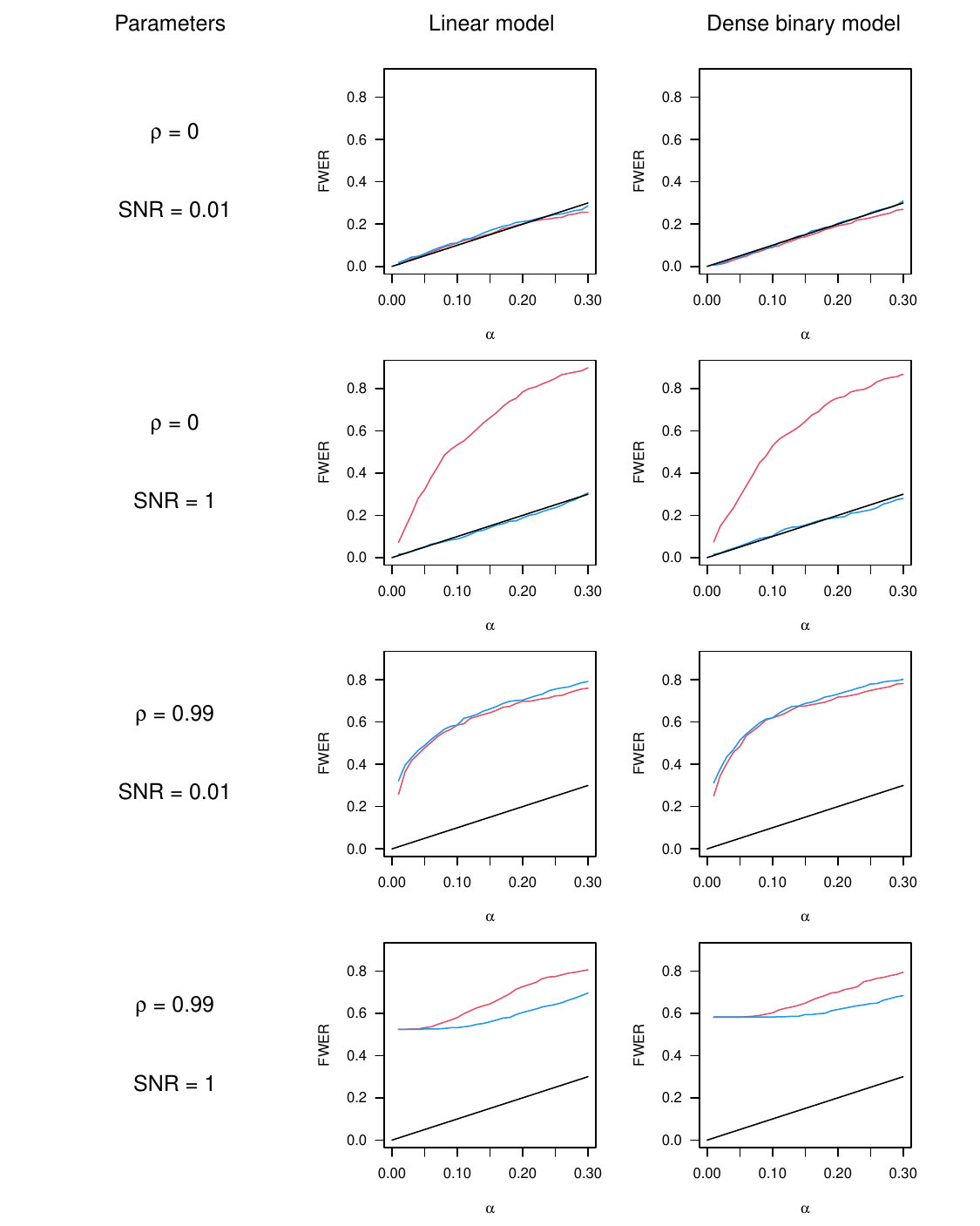}}
    \caption{FWER of the variable selection procedure with thresholding (blue) or \textit{ForwardStop} (red) as a function of $\alpha$ in 8 scenarios of linear model or dense binary model with 10 active variables.}
    \label{fig:procedure-simcal-FWER-1}
\end{figure}

\begin{figure}
    \centering
    \makebox[\textwidth][c]{\includegraphics[width=\textwidth]{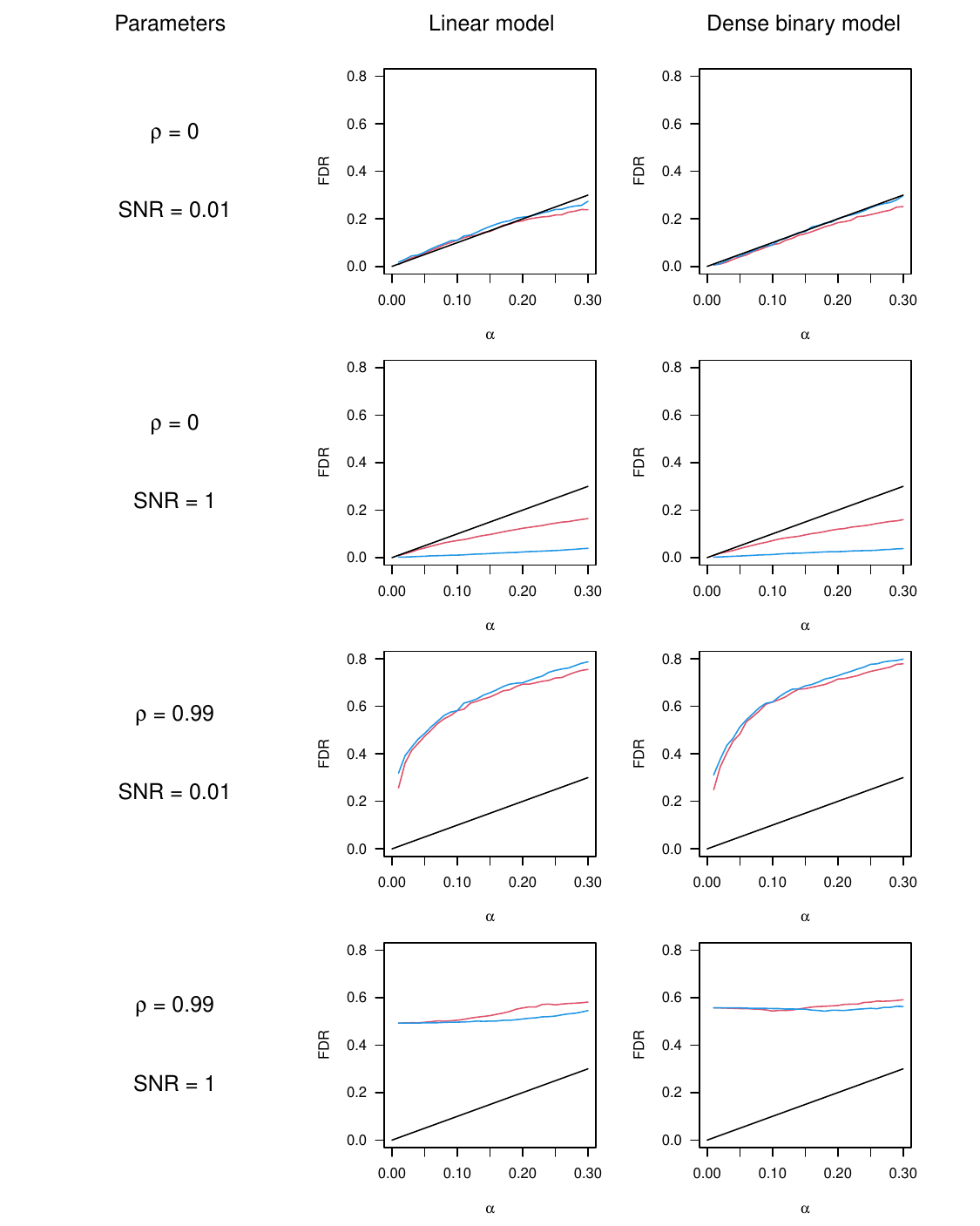}}
    \caption{FDR of the variable selection procedure with thresholding (blue) or \textit{ForwardStop} (red) as a function of $\alpha$ in 8 scenarios of linear model or dense binary model with 10 active variables.}
    \label{fig:procedure-simcal-FDR-1}
\end{figure}

Although theoretical results on controlling the FWER or FDR are valid only in the linear model, the shape of these curves varies little across model types. In scenarios where it is easiest to detect a signal, i.e., those with high SNR and especially among those, the uncorrelated scenarios, it is observed that the ForwardStop criterion is significantly less conservative than thresholding at an equal $\alpha$, with higher FWER and FDR. This is explained by a substantial number of small p-values among the first values of the sequence $(p_k)$, corresponding to easily detected variables, which, by averaging, pull down the quantity $p_k^{FS}$ at higher $k$, making it significantly lower than $p_k$.
\\
In all uncorrelated scenarios, as expected, FWER control is very close to $\alpha$ with the thresholding criterion (blue curves in Figure \ref{fig:procedure-simcal-FWER-1} are close to or below the bisectors). Additionally, in these scenarios, we observe FDR control at a level close to or lower than $\alpha$ with the ForwardStop criterion (red curves in Figure \ref{fig:procedure-simcal-FDR-1}). Due to the more conservative nature of thresholding compared to ForwardStop, thresholding also controls the FDR below $\alpha$ in uncorrelated scenarios, but ForwardStop does not control the FWER. The gain in sensitivity of ForwardStop compared to thresholding is modest, reaching a maximum of $0.127$ across all 252 scenarios.
\\
On the other hand, in scenarios with many active variables and strong correlation among covariates, neither the FWER nor the FDR is controlled by either halting criterion. The strong correlation between active and inactive variables often allows a nominally inactive variable to be selected by Lasso before the corresponding active variable. The selected variable can then have a statistically significant association with the response which is a manifestation of the actual association between the active variable and the response. This cannot be attributed to the nominally active variable, because it is not known and not detected by the Lasso. The status of these "false positives" carrying a statistical signal has been discussed in the literature \citep{gsell_sequential_2016}.

\subsubsection{Comparison with Lockhart et al.'s Covariance Test} \label{comparison covtest}

To evaluate the sensitivity of the variable selection procedure via simulation-calibration, it is useful to compare it to another selection method aiming for the same objectives. This competitor is the covariance test (CovTest) by \cite{lockhart_significance_2014}, which, like the simulation-calibration test, measures the significance of variables entering the Lasso path by assigning each one a p-value.
\\
Due to reasons of implementation of CovTest, we focused on the 63 linear model scenarios. We used CovTest with the equivalent of algorithm \ref{algorithme simcal itere}, with the thresholding halting criterion at $\alpha = 0.05$, which is supposed to control the FWER at this level. Figure \ref{fig:comparaison-covtest} shows the FWER and sensitivity of both methods across all linear model scenarios.

\begin{figure}
    \centering
    \makebox[\textwidth][c]{\includegraphics[width=0.9\textwidth]{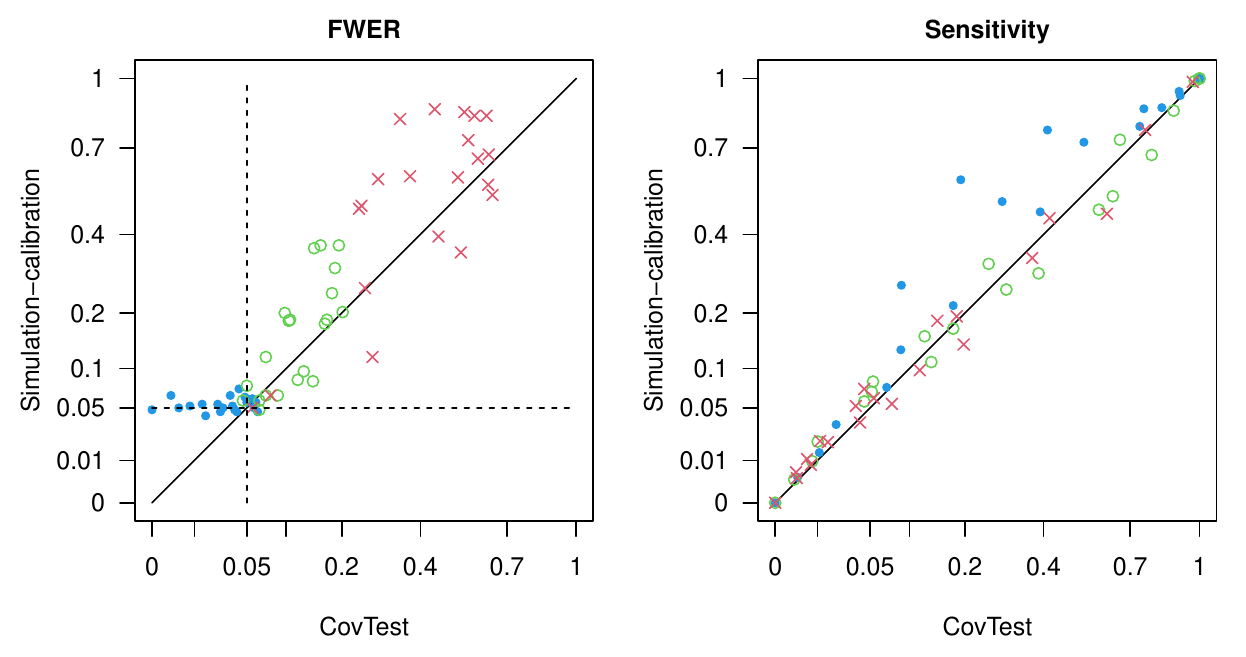}}
    \caption{Performance comparison of selection procedures with thresholding of CovTest and simulation-calibration test at $\alpha = 0.05$ across the 63 linear models (\textcolor{paletteBlue}{$\bullet$} $\rho = 0$, \textcolor{paletteGreen}{$\circ$} $\rho = 0.9$, \textcolor{paletteRed}{$\times$} $\rho = 0.99$). The position of a point indicates the performance of both methods on the same scenario, with the bisectors representing equal performance. The scale is quadratic.}
    \label{fig:comparaison-covtest}
\end{figure}

We observe a failure to control the FWER in correlated scenarios by both CovTest and simulation-calibration (points \textcolor{paletteGreen}{$\circ$} and \textcolor{paletteRed}{$\times$} in figure \ref{fig:comparaison-covtest}). In these scenarios, simulation-calibration tends to produce false positives more frequently than the CovTest when the signal-to-noise ratio (SNR) is low, but less frequently when it is high (see the impact of SNR in figure 6 of the supplementary material). However, the FWER is controlled at $5\%$ or less in the 21 linear scenarios with no correlation, both by the CovTest and by simulation-calibration (points \textcolor{paletteBlue}{$\bullet$}), but with a clear difference between the two methods.
\\
The FWER observed with the CovTest selection procedure falls well below its nominal level in certain scenarios, even reaching 0 --- meaning no false positives observed in 500 simulations --- in the scenario with 10 regressors, $\mathrm{SNR} = 0.3$.
\\
In contrast, the FWERs observed with the simulation-calibration procedure across the 21 uncorrelated scenarios all hover close to their nominal levels, ranging from $0.042$ to $0.072$. This implies that the number of simulated data sets on which the procedure produces at least one false positive ranges between $21$ and $36$ out of $500$. Given the Bonferroni correction, this is consistent with the hypothesis that the true FWER is equal to $\alpha$ in each of these scenarios, i.e., that the number of data sets exhibiting at least one false positive in each scenario follows a $\mathrm{Binomial}(500, 0.05)$ distribution. Indeed, the p-value associated with the largest of the 21 observed FWERs is $0.0196 = 0.412 / 21$.
\\
This stronger conservatism of the CovTest procedure in uncorrelated scenarios results in notable gaps in sensitivity to the advantage of the simulation-calibration procedure. The sensitivity difference is positive or zero in each of these 21 scenarios, exceeds $0.1$ in 6 of them (all with at least 5 active regressors and an SNR of at least $0.1$), and reaches a maximum of $0.388$ in the scenario with 10 regressors, $\mathrm{SNR} = 0.3$. Therefore, the simulation-calibration test represents an improvement over the CovTest under the ideal condition of no correlation between covariates, with substantial gains in sensitivity allowed by a controlled increase in FWER that does not significantly exceed its nominal level.

\section{Application to pharmacovigilance data} \label{BNPV}

We illustrated the sequential variable selection procedure by using it on data from the national pharmacovigilance database (BNPV). We used the same data preprocessing as described in \cite{courtois_new_2021}, producing a database of $n = \numprint{452914}$ spontaneous notifications of adverse drug reactions from January 1, 2000, to December 29, 2017, with $6617$ distinct adverse events (coded at the Preferred Term level of the Medical Dictionary for Regulatory Activities, MedDRA) and $p = 1692$ distinct drugs (coded at the 5\textsuperscript{th} level of the Anatomical Therapeutic Chemical hierarchy) reported at least 10 times. We focused on a binary outcome, the adverse event "Drug-Induced Liver Injury" (DILI). It is one of the most frequent adverse events with $25187$ occurrences, accounting for $5.56\%$ of all spontaneous notifications. We used a logistic regression model on the drug exposures, which are binary covariates. To reduce computation time, we used only $N = 50$ simulations in the computation of empirical p-values.
\\
When a variable is selected with a positive coefficient estimate, we considered it to be a pharmacovigilance signal. To assess the performance of our method, as in \cite{courtois_new_2021}, we used the reference set DILIrank of known pharmacovigilance signals related to drug-induced liver injuries \cite{chen_dilirank_2016}. It includes 203 negative controls (drugs known not to be associated with DILI) and 133 positives (drugs known to be associated with DILI).
\\
Given initial results, it appeared necessary to preprocess the data before performing the Lasso. Figure 7 in the supplementary material shows that the p-values estimated by simulation-calibration are zero for the first 14 variables on the Lasso path, then distributed around $0.1$. This sudden change coincides with a peculiarity in the correlation structure of the exposure matrix $X$: the 14\textsuperscript{th} exposure selected, trimethoprim (ATC code J01EA01), is exceptionally correlated ($\rho = 0.9998$) with another exposure, sulfamethoxazole (J01EC01), a drug with which trimethoprim is almost always co-prescribed.

\begin{figure}
    \centering
    \makebox[\textwidth][c]{\includegraphics[width=0.9\textwidth]{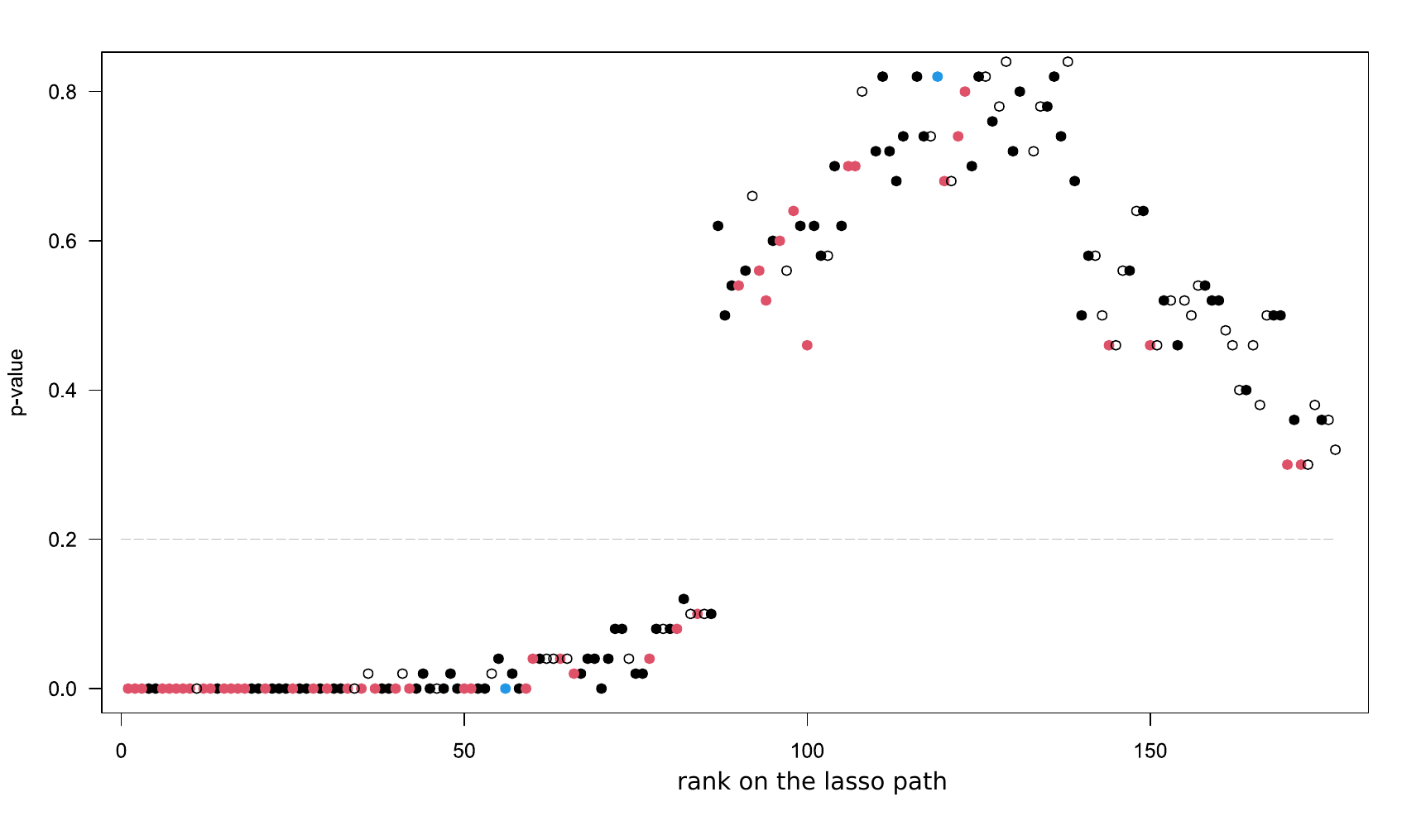}}
    \caption{p-values estimated by simulation-calibration along the Lasso path after removing the strongest correlations, with color indicating the status of signals according to DILIrank. \textcolor{paletteRed}{$\bullet$}: positive association known as such; \textcolor{paletteBlue}{$\bullet$}: positive association contradicted by DILIrank; $\bullet$: unlisted positive association; $\circ$: negative association (no signal).}
    \label{fig:simcal-BNPV-nocor}
\end{figure}

\begin{figure}
    \centering
    \makebox[\textwidth][c]{\includegraphics[width=0.9\textwidth]{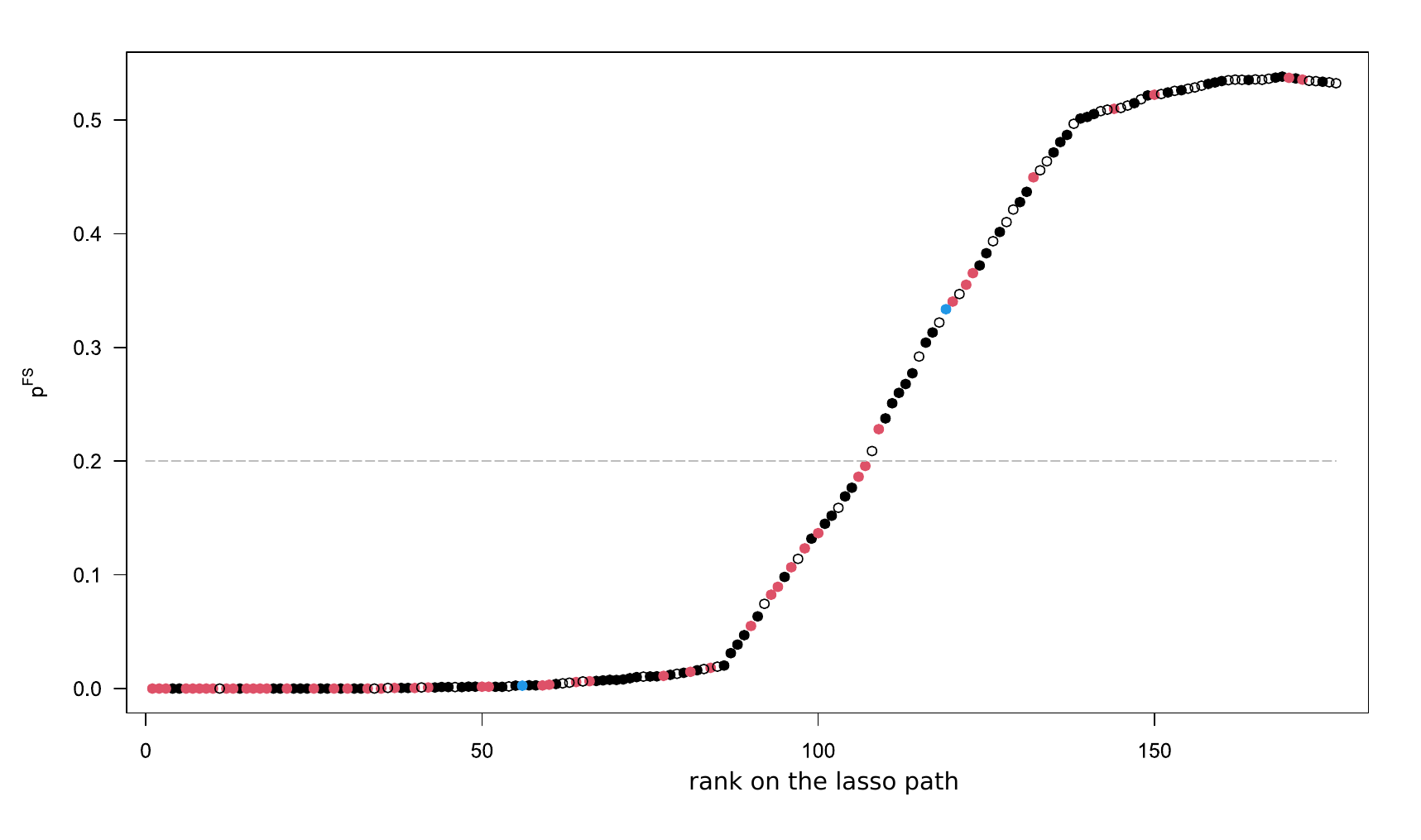}}
    \caption{Quantities of interest from the \textit{ForwardStop} criterion, $p^{FS}$ (see section \ref{ForwardStop}), derived from the p-values estimated by simulation-calibration after removing the strongest correlations. The color indicates the status of the signals as on figure \ref{fig:simcal-BNPV-nocor}.}
    \label{fig:simcal-BNPV-nocor-FS}
\end{figure}

We removed from the matrix $X$ the 6 covariables that had a correlation of at least $0.9$ with any of the variables on the Lasso path (computed before this removal), and re-used the Lasso and the sequence of simulation-calibration tests on these purged data ($p = 1686$). The Lasso yields identical results since the removed variables were not selected by it. However, p-values estimated by simulation-calibration (Figure \ref{fig:simcal-BNPV-nocor}) are lower starting from the 15\textsuperscript{th} variable, and no longer show a regime change at that point. They do show a comparable regime change further along the Lasso path, after the p-value associated with the 86\textsuperscript{th} selected variable (P01BDXX, diaminopyrimidines). This variable is also correlated with another variable in the database (P01BD01, pyrimethamine), but at a lower level: $\rho = 0.307$, similar to the other correlations involving selected variables.
\\
Figure \ref{fig:simcal-BNPV-nocor-FS} shows that using ForwardStop instead of thresholding, in addition to resulting in less conservative selection at the same $\alpha$, has a stabilizing effect: while the empirical p-values exhibit fluctuations due to the limited number of simulations from which they are estimated, the quantity of interest $p^{FS}$ on which ForwardStop is based presents a much smoother profile, being calculated by averaging from the empirical p-values.

\begin{table}
\centering
\caption{\label{tab:BNPV simcal} Performance of variable selection by simulation-calibration on BNPV data in terms of number of pharmacovigilance signals (variables with a positive association with DILI), false discovery proportion (FDP), specificity and sensitivity.}
\begin{tabular}{ l  r  r  r  r  r  r } \hline
    Halting & \multirow{3}{*}{Signals} & Signals & False & FDR & Specificity & Sensitivity \\
    criterion &  & with known & positives & (\%) & (\%) & (\%) \\
    &  & status &  &  &  &  \\  \hline
    Thresholding  & \multirow{2}{*}{73} & \multirow{2}{*}{33} & \multirow{2}{*}{1} & \multirow{2}{*}{3.0} & \multirow{2}{*}{99.5} & \multirow{2}{*}{24.1} \\
    at $\alpha = 0.2$ &  &  &  &  &  &  \\
    ForwardStop  & \multirow{2}{*}{91} & \multirow{2}{*}{41} & \multirow{2}{*}{1} & \multirow{2}{*}{2.4} & \multirow{2}{*}{99.5} & \multirow{2}{*}{30.0} \\
    at $\alpha = 0.2$ &  &  &  &  &  &  \\ \hline
\end{tabular}
\end{table}

Table \ref{tab:BNPV simcal} shows the performances of the variable selection procedure --- with the thresholding or ForwardStop halting criterion, at the same level $\alpha = 0.2$ --- calculated from the status of drug exposures listed in the DILIrank reference set \citep{chen_dilirank_2016}. Only exposures where Lasso estimates a positive association with DILI are considered signals. The status of the exposures is also indicated in figures \ref{fig:simcal-BNPV-nocor} and \ref{fig:simcal-BNPV-nocor-FS} allowing visualization of performances at other levels of $\alpha$. Both approaches are markedly more conservative than the information criteria (including ForwardStop, which itself is less conservative than thresholding). This allows them to achieve a lower false positive rate at the expense of lower sensitivity. Daunorubicin (L01DB02) is the sole false positive according to DILIrank.

\section{Discussion}
We have proposed a test for the significance of variables entering the Lasso path, which <e use sequentially to select a model. It tests the null hypothesis $H_0 (A)$ which states that a known set of variables---$A$---includes all the active variables, and it focuses on $\lambda_A$, the largest value of $\lambda$ at which a variable not belonging to $A$ is selected. Rejecting the test means selecting this variable.
\\
We built the test to circumvent the difficulty of using $\lambda_A$ directly as a test statistic. Just as a p-value is the probability under the null hypothesis that the test statistic exceeds its observed value, the statistic we consider, $p_A$, is a conditional probability that $\lambda_A$ exceeds its observed value: conditional on the correlation structure relating the response $Y$ to the already selected variables $X_A$, given by $\widehat{\theta_A}$. This is why we call it a conditional p-value. We estimate it by the Monte Carlo method provided by algorithm \ref{linear simcal algorithm}: simulation-calibration. It involves simulating response vectors under the null hypothesis and then calibrating them on the condition on $\widehat{\theta_A}$. This distributes the response vectors following the conjunction of the null hypothesis and this condition, which allows using them to estimate the conditional p-value.
\\
This can be seen as a generalization of the permutation selection by \cite{sabourin_permutation_2015}, which simulates a population of response vectors with the same distribution as the observed $y$, and considers the population of the $\lambda$ at which a variable enters the Lasso path computed on these simulated vectors. In permutation selection, the focus in on the $\lambda$ at which the very first variable enters the Lasso path, i.e., in our notation, $\lambda_\emptyset$. Permutation is also a way to impose conserving $\widehat{\theta_\emptyset} = \left(\widehat{\beta_0}, \hat{\sigma} \right)$, the empirical mean and standard deviation of $y$. Permutation selection resembles the test of $H_0 (\emptyset)$ by simulation-calibration in these aspects. Since permutation selection retains the median of $\lambda_\emptyset$, using it to decide only on the selection or not of the first variable on the Lasso path would essentially amount to conducting this test at $\alpha = 0.5$ (the only difference being in the method of calibrated simulation, permutation, or post-simulation calibration of arbitrary vectors).
\\
We have proven in sections \ref{conditional pvalue} to \ref{prop simcal} the validity of our method in the case of the linear model. The conditional p-value $p_A (Y)$ under the null hypothesis follows a distribution which stochastically dominates the uniform distribution over $[0,1]$ (lemma \ref{inequality pAY}), and its estimator $\widehat{p_A} (Y)$ is unbiased, consistent with the number $N$ of simulations chosen to perform (lemma \ref{binomial distribution} and its consequences), and itself dominates the uniform distribution over the discrete set of values it can take (lemma \ref{discrete uniform}). This means that rejecting the null hypothesis based on a threshold $\alpha$ on $\widehat{p_A} (Y)$ controls the type I error rate at this level $\alpha$, with a small residual term. From a variable selection perspective, this type I error means the selection (necessarily erroneous) of an additional variable while all active variables are already selected, i.e., belong to $A$. Furthermore, we have proven (section \ref{extended theorem} and supplementary matereial) that even when some active variables have not been selected, the risk of selecting an inactive variable by simulation-calibration is controlled at this same level. This latter result requires orthogonality between active and inactive covariates. The control of the error under $H_0 (A)$ does not depend on any assumptions about the correlation structure.
\\
In discrete generalized linear models (binary and Poisson), the theoretical framework is different since there is no conditional p-value that follows or dominates the continuous uniform distribution over $[0,1]$. Nevertheless, the production of an empirical p-value by simulation-calibration adapts well to these models. Theoretically, $\widehat{p_A} (Y)$ is seen as the estimator of a probability whose definition approaches that of the conditional p-value (section \ref{nonlinear problem statement}). Practically, the "calibration" part of the algorithm becomes more complex: it is iterative and stochastic (algorithm \ref{nonlinear iterative calibration}) while linear calibration was a simple affine function. Although it is not guaranteed by a theorem, the simulation study (section \ref{etude simu simcal H0}) shows that the distribution of empirical p-values is indistinguishable from the uniform distribution in the majority of non-linear simulation scenarios and deviates slightly in others, without any scenario of failure to control the type I error at usual levels of $\alpha$.
\\
Simulation-calibration has the disadvantage of requiring a computation time that can be significant under certain circumstances. It involves, for each of the $N$ simulated vectors, adjusting the unpenalized model restricted to $A$ in order to perform the calibration, then applying the Lasso to the entire set of covariates to the calibrated vector. In nonlinear models, calibration is itself iterative and requires repeatedly readjusting the unpenalized restricted model. In the application to the BNPV data, the conditions for a very long computation time were met: large data ($\numprint{452914}$ by $1686$), a binary model, and especially an $A$ set of selected variables that reaches a significant size (over 100). It was therefore necessary to repeatedly adjust relatively large-dimensional logistic models without penalty. By contrast, the simulation study was carried out with a maximum of 10 active variables, so the number of selected variables did not exceed this order of magnitude, which limited the computation time.
\\
Our test is comparable, in its objectives, to the covariance test (CovTest) of \cite{lockhart_significance_2014}, which is also a test of significance for a variable preselected by the Lasso. We focus on the distribution of $\lambda_A$ alone and not, like Lockhart et al., on the evolution of the Lasso between two consecutive values of $\lambda$ where a variable enters the Lasso. Therefore, we avoid a situation of low power of the covariance test: when two active variables enter the Lasso at values of $\lambda$ close from each other, the CovTest's test statistic is the difference between the same quantity computed on Lasso results which differ only a little; this can lead artificially to selecting none of the two active variables even if both are significantly associated with the outcome. This could explain simulation-calibration's better sensitivity in the absence of correlation between covariates that we observe in section \ref{comparison covtest}. In the presence of correlation, however, the performance of simulation-calibration is not always better than that of the selection by CovTest, with lower FWER at equal sensitivity in high SNR scenarios, but higher FWER at equal sensitivity in low SNR scenarios (Figure \ref{fig:comparaison-covtest} and, in the supplementary material, Figure 6).
\\
Using the procedure on data from the French national pharmacovigilance database has highlighted an undesirable phenomenon that could help explain the poorer performance in some cases of correlation. We observe in Figure \ref{fig:simcal-BNPV-nocor} a change of pattern in the sequence of estimated p-values where, once a certain variable $j_0$ belongs to $A$, $\widehat{p_A} (y)$ is always relatively high. The detail of the Lasso performed in the simulation-calibrations provides an explanation for this phenomenon. We observe that when it occurs, there is a variable $j_1 \notin A$ correlated with $j_0$ and possibly with other variables in $A$ such that, among the population of $\lambda_A (y^{(l)})$ obtained by simulation-calibration (see algorithm \ref{nonlinear simcal algorithm}), a significant proportion are high values associated with an early selection of $j_1$ by the Lasso applied to $y^{(l)}$. Indeed, calibrating the $y^{(l)}$ on one or more associations with variables correlated with $j_1$ can lead to an association with $j_1$ itself, which is captured by the Lasso. These high $\lambda_A (y^{(l)})$ lead to the estimation of a high $\widehat{p_A} (y)$.
\\
This phenomenon reflects the influence of $j_A$ on the distribution of $\lambda_A$ under the null hypothesis. It might be more relevant, from the perspective of the test's power, to compare the observed $\lambda (y)$, not as we have done with the distribution of $\lambda_A (Y)$ conditional on $\widehat{\theta_A} (Y)$, but with its distribution under a double conditioning by $\widehat{\theta_A} (Y), j_A (Y)$, or more generally taking into account $j_A (Y)$. In practice, this would involve introducing a weighting in the calculation by averaging $\widehat{p_A} (Y)$.
\\
This phenomenon results, when $j_0 \in A$, in a loss of power in the selection of all variables, including those that are not correlated with $j_0$ or $j_1$. This loss of power in certain cases of correlation could offset the tendency of the test by simulation-calibration, observed in simulations without correlation, to be more powerful but less conservative than the CovTest which would explain why in some correlated scenarios, variable selection by simulation-calibration is less conservative than that by CovTest while it has practically the same sensitivity.
\\
Thesimulation-calibration procedure's lower FWER compared to that by CovTest in high-SNR correlated scenarios could be explained by the fact that due to the high SNR, it is more common for all active variables to be selected before inactive variables on the Lasso path. When this is the case, $H_0 (A)$ is verified (thus the risk of false positive is controlled) in all tests of the iterative procedure that are likely to produce a false positive. For the risk control to fail, it is necessary both for the active variables to be mixed with inactive variables on the Lasso path -- for there to be an iteration where an unverified $H_0 (A)$ is tested which can nevertheless lead to the selection of an inactive variable --, and for active variables to be correlated with inactive variables -- so that the extended theorem of control of the selection error (section \ref{extended theorem}) does not apply.

\bibliography{biblio.bib}

\end{document}


\maketitle

\section{Proofs}

Complete linear model:
\begin{equation} \label{linmodel}
y = \beta_0 + X\beta + \epsilon, \ \ \ \ \epsilon \sim \mathcal{N} \left(0,\sigma^2 I_n \right)
\end{equation}
Restricted linear model:
\begin{equation} \label{modelA}
y = X_A \beta_A + \epsilon_A, \ \ \ \ \epsilon_A \sim \mathcal{N}(0,\sigma_A^2 I_n)
\end{equation}
\begin{definition}[$p_A$ statistic] \label{pAy}
$$p_A (y) = \mathrm{P}_{H_0 (A)} \left(\lambda_A (Y) \geq  \lambda_A (y) | \widehat{\theta_A} (Y) = \widehat{\theta_A} (y) \right).$$
\end{definition}

To know the distribution of $p_A (y)$, we need the following lemma, which is a general result about the cumulative distribution function of a real-valued random variable.

\begin{lemma} \label{inequality cdf}
Let $U$ be a real-valued random variable and $F$ its cumulative distribution function. Then:
$$\forall t \in [0,1], \mathrm{P} \left(F(U) \leq t \right) \leq t,$$

and if $U$ follows a continuous distribution, then the probability is equal to $t$.
\end{lemma}

\begin{proof}
First, let's prove the inequality. Let $t \in [0,1]$ and let $I = \{ u \in \mathbb{R} : F(u) \leq t \}$. We want to prove that $\mathrm{P} (U \in I) \leq t$. Since $F$ is increasing, $I$ is an interval not bounded on the left.

The inequality is obviously true if $t = 1$. Therefore, we can assume $t < 1$. As $\lim_{\substack{+ \infty}} F = 1$, there exists $u \in \mathbb{R}$ such that $F(u) > t$.

If $t = 0$ and $F$ is strictly positive, then $\forall u, F(u) > t$ so $\mathrm{P} \left(F(U) \leq t \right) = 0$ thus the inequality is true. Therefore, we can assume that either $\exists u \in \mathbb{R}$ such that $F(u) = 0$, or $t > 0$. In the first case, such a $u$ is an element of $I$. In the second case, as $\lim_{\substack{- \infty}} F = 0$, there exists $u \in \mathbb{R}$ such that $F(u) \leq t$, i.e., an element of $I$.

Outside the trivial cases, $I$ is therefore a bounded non-empty set. It thus has a real upper bound, denoted $u^{+}$. $I \subset ]- \infty, u^{+}]$ so
$$\mathrm{P} (U \in I) \leq \mathrm{P} (u \leq u^{+}) = F(u^{+}).$$

So if $F(u^{+}) \leq t$, then indeed $\mathrm{P} (U \in I) \leq t$.

If, on the other hand, $F(u^{+}) > t$, then $u^{+} \notin I$ by definition of $I$. Since $u^{+}$ is an upper bound, $I = ]- \infty, u^{+}[$ and there exists an increasing sequence ${(u_p)}_{p \in \mathbb{N}}$ of elements of $I$ such that $\lim_{\substack{p \to \infty}} u_p = u^{+}$. Thus $I = \cup_{\substack{p}} ]- \infty, u_p]$. Being an increasing union of measurable sets, we have
\begin{eqnarray}
\mathrm{P} (U \in I) &=& \lim_{p \to \infty} \mathrm{P} (U \in ]- \infty, u_p]) \nonumber \\
&=& \lim_{p \to \infty} F(u_p). \nonumber
\end{eqnarray}
Since by membership in $I$, $\forall p, F(u_p) \leq t$, thus by taking the limit:
$$\mathrm{P} (U \in I) \leq t.$$

With the inequality proven, let's deduce the equality in the continuous case. Denoting $G$ as the cumulative distribution function of $-U$, we have $\forall u \in \mathbb{R}$:
\begin{eqnarray}
G (-u) &=& \mathrm{P} (- U \leq -u) \nonumber \\
&=& \mathrm{P} (U \geq u) \nonumber \\
&=& \mathrm{P} (U > u) + \mathrm{P} (U = u) \nonumber \\
&=& 1 - F(u) + \mathrm{P} (U = u). \nonumber
\end{eqnarray}

If the distribution of $U$ is continuous, then $\mathrm{P} (U = u) = 0$ so $G(-u) = 1 - F(u)$. Now apply the lemma's inequality to the variable $-U$ and the threshold $1-t$:
\begin{eqnarray}
\mathrm{P} \left(G(-U) \leq 1-t \right) &\leq& 1-t \nonumber \\
\mathrm{P} \left(1 - F(U) \leq 1-t \right) &\leq& 1-t \nonumber \\
\mathrm{P} \left(F(U) \geq t \right) &\leq& 1-t \nonumber \\
\mathrm{P} \left(F(U) < t \right) &\geq& t \ \ \text{(complement event)} \nonumber \\
\mathrm{P} \left(F(U) \leq t \right) &\geq& t \ \ \text{(event containing the previous one)} \nonumber
\end{eqnarray}
Which, combined with the lemma's inequality, gives $\mathrm{P} \left(F(U) \leq t \right) = t$.

\end{proof}

\begin{lemma} \label{inequality pAY}
Under $H_0 (A)$, the distribution of $p_A (Y)$ stochastically dominates the uniform distribution on $[0,1]$:
$$\forall t \in [0,1], \mathrm{P}_{H_0 (A)} \left(p_A(Y) \leq t \right) \leq t.$$
\end{lemma}

\begin{proof}
Let's decompose the quantity to be bounded by integration over $\Theta_A$, the set of values that $\widehat{\theta_A} (Y)$ can take:
\begin{eqnarray} \label{integral pAY}
\forall t\in [0,1], \ &\mathrm{P}_{H_0 (A)}& (p_A (Y) \leq t) = \\
&\int_{\theta_{A}' \in \Theta_A}& \mathrm{P}_{H_0 (A)} \left(p_A (Y) \leq t | \widehat{\theta_A} (Y) = \theta_{A}' \right) d\mathrm{P}_{H_0 (A)} \left(\widehat{\theta_A} (Y) = \theta_{A}' \right). \nonumber
\end{eqnarray}

For all parameter vectors $\theta_{A}'$, let $F_{\theta_{A}'}$ denote the cumulative distribution function of the distribution of $- \lambda_A (Y)$ conditionally on $\widehat{\theta_A} (Y) = \theta_{A}'$. Then for all $y$,
$$ p_A (y) = \mathrm{P}_{H_0 (A)} \left(- \lambda_A (Y) \leq  - \lambda_A (y) | \widehat{\theta_A} (Y) = \widehat{\theta_A} (y) \right) = F_{\widehat{\theta_A} (y)} \left( - \lambda_A (y) \right).$$
We now apply lemma \ref{inequality cdf} by taking the conditional distribution of $- \lambda_A (Y)$ as the distribution of $U$. The obtained probability is thus conditional:
\begin{equation} \label{conditional inequality}
\forall t \in [0,1], \forall \theta_{A}' \in \Theta_A, \ \mathrm{P}_{H_0 (A)} \left(p_A (Y) \leq t | \widehat{\theta_A} (Y) = \theta_{A}' \right) \leq t.
\end{equation}
This inequality is substituted into the integral \ref{integral pAY}:

$$\mathrm{P}_{H_0 (A)} (p_A (Y) \leq t) \leq \int_{\theta_{A}' \in \Theta_A} t \ d\mathrm{P}_{H_0 (A)} \left( \widehat{\theta_A} (Y) = \theta_{A}' \right) = t.$$
\end{proof}

\begin{lemma} \label{calibration equality}
For all $\theta_{A}^{(1)} \in \Theta_A^*, \theta_{A}^{(2)} \in \Theta_A$, for all $y^{(1)} \in \mathbb{R}^n$ satisfying $\widehat{\theta_A} (y^{(1)}) = \theta_{A}^{(1)}$, we have:
$$\widehat{\theta_A} \left( \mathrm{cal}_{\theta_{A}^{(1)} \rightarrow \theta_{A}^{(2)}} (y^{(1)}) \right) = \theta_{A}^{(2)}.$$
\end{lemma}
\begin{proof}
Let $\left( \theta_{A}^{(1)}, \theta_{A}^{(2)}, y^{(1)} \right) \in \Theta_A^* \times \Theta_A \times \mathbb{R}^n$ with $\widehat{\theta_A} (y^{(1)}) = \theta_{A}^{(1)}$, i.e., $\widehat{\beta_A} (y^{(1)}) = \beta_{A}^{(1)}$ and $\widehat{\sigma_A} (y^{(1)}) = \sigma_{A}^{(1)}$.

Let $y^{(2)} = \mathrm{cal}_{\theta_{A}^{(1)} \rightarrow \theta_{A}^{(2)}} (y^{(1)})$. Then:
\begin{eqnarray}
\widehat{\beta_A} (y^{(2)}) &=& (X_A^T X_A)^{-1} X_A^T y^{(2)} \nonumber \\
&=&  (X_A^T X_A)^{-1} X_A^T \left( X_A \beta_{A}^{(2)} + \frac{\sigma_A^{(2)}}{\sigma_A^{(1)}} \left(y^{(1)} - X_A \beta_{A}^{(1)} \right) \right) \nonumber \\
&=& (X_A^T X_A)^{-1} X_A^T X_A \beta_{A}^{(2)} \nonumber \\
&+& \frac{\sigma_A^{(2)}}{\sigma_A^{(1)}} \left((X_A^T X_A)^{-1} X_A^T y^{(1)} - (X_A^T X_A)^{-1} X_A^T X_A \beta_{A}^{(1)} \right) \nonumber \\
\widehat{\beta_A} (y^{(2)}) &=& \beta_{A}^{(2)} + \frac{\sigma_A^{(2)}}{\sigma_A^{(1)}} \left(\widehat{\beta_A} (y^{(1)}) - \beta_{A}^{(1)} \right). \nonumber
\end{eqnarray}
$\widehat{\beta_A} (y^{(1)}) = \beta_{A}^{(1)}$ so $\widehat{\beta_A} (y^{(2)}) = \beta_{A}^{(2)}$. Furthermore:
\begin{eqnarray}
\widehat{\sigma_A}^2 (y^{(2)}) &=&  \frac{1}{n} {\left| \left| y^{(2)} - X_A \widehat{\beta_A} (y^{(2)}) \right| \right|}^2 \nonumber \\
&=&  \frac{1}{n} {\left| \left| X_A \left(\beta_{A}^{(2)} - \widehat{\beta_A} (y^{(2)}) \right) + \frac{\sigma_A^{(2)}}{\sigma_A^{(1)}} \left(y^{(1)} - X_A \beta_{A}^{(1)} \right) \right| \right|}^2 \nonumber \\
&=& \left( \frac{\sigma_A^{(2)}}{\sigma_A^{(1)}} \right)^2 \frac{1}{n} {\left| \left| y^{(1)} - X_A \beta_{A}^{(1)} \right| \right|}^2 \ \ \ \text{because} \ \widehat{\beta_A} (y^{(2)}) = \beta_{A}^{(2)} \nonumber \\
&=& \left( \frac{\sigma_A^{(2)}}{\sigma_A^{(1)}} \right)^2 \widehat{\sigma_A}^2 (y^{(1)}) \nonumber \\
\widehat{\sigma_A}^2 (y^{(2)}) &=& {\sigma_A^{(2)}}^2 \ \ \ \text{because} \ \widehat{\sigma_A} (y^{(1)}) = \sigma_{A}^{(1)}. \nonumber
\end{eqnarray}
Thus $\widehat{\theta_A} (y^{(2)}) = \theta_{A}^{(2)}$.
\end{proof}

\begin{corollary} \label{uniform distribution pAY}
Under the continuity assumption, $p_A(Y)$ follows a uniform distribution on $[0,1]$.
\end{corollary}
\begin{proof}
The proof is identical to that of Lemma \ref{inequality pAY} with two adaptations:
\begin{itemize}
\item{The integrals are over $\Theta_A^*$ instead of $\Theta_A$, which yields the same result since $\mathrm{P} \left( \widehat{\theta_A} (Y) \notin \Theta_A^* \right) = \mathrm{P} ( \widehat{\sigma_A} (Y) = 0 ) = 0 $.}
\item{For all $\theta_{A}' \in \Theta_A^*$, the distribution of $- \lambda_A (Y)$ conditionally on $\widehat{\theta_A} (Y) = \theta_{A}'$ is continuous, so we are in the equality case of Lemma \ref{inequality cdf}. The inequality \eqref{conditional inequality} is thus an equality. It is preserved under integration.}
\end{itemize}
\end{proof}

\paragraph{Discussion on the continuity assumption}

The continuity assumption is stated on $\Theta_A^*$ and not on $\Theta_A$ because parameter vectors of the type $\theta_{A}' = (\beta_{A}', 0)$ (with standard deviation zero) lead to a non-continuous distribution of $\lambda_A$ conditionally on $\widehat{\theta_A} (Y) = \theta_{A}'$. That is because conditioning in this way implies that $\widehat{\sigma_A} (Y) = 0$, therefore $Y$ is a linear combination of the variables in $A$ (making it impossible to locally vary $Y$ while maintaining the conditioning). Under this conditioning, there is a non-zero probability that the Lasso selects no variables outside of $A$, resulting in $\lambda_A (Y) = 0$.
\\
Moreover, it is possible that the continuity assumption is not satisfied. For example, if there exists a variable $X_{j'}$ that is a linear combination of the $X_{j}, j \in A$, then its covariance with any linear combination of $Y$ and the $X_{j}, j \in A$ is uniquely determined by $\widehat{\theta_A}$. If no other variable besides $X_{j'}$ and those in $A$ appears on the Lasso path before $X_{j'}$, then the $\lambda$ at which $X_{j'}$ appears on the Lasso path is determined by these covariances. Therefore, there exists a value of $\lambda$, a deterministic function of $\theta_{A}'$, that $\lambda_A$ takes with a non-zero probability conditionally on $\widehat{\theta_A} (Y) = \theta_{A}'$ (since the probability that $j'$ is the first variable selected outside of $A$ is non-zero). In such cases, the conditional distribution of $\lambda_A (Y)$ is not continuous.

\begin{lemma} \label{Y uniforme}
Under $H_0 (A)$, the distribution of $Y$ conditional on $\widehat{\theta_A} (Y) = \widehat{\theta_A} (y)$ is the uniform distribution on the set $S \left( \widehat{\theta_A} (y) \right) = \left\{ y' \in \mathbb{R}^n \mid \widehat{\theta_A} (y') =\widehat{\theta_A} (y) \right\}$.
\end{lemma}

\begin{proof}
Let $p_Y$ be the probability distribution of $Y$ and let $y' \in \mathbb{R}^n$. Then, according to model \ref{linmodel}:
$$\log p_Y (y') = -\frac{1}{2{\sigma}^2} (y' - X\beta)^T (y' - X\beta) + \text{Constant}.$$

Suppose $H_0 (A)$ is true. Then $Y$ follows model \ref{modelA}:
\begin{eqnarray}
\log p_Y (y') &=& -\frac{1}{2{\sigma}^2} (y'-X_A \beta_A )^T (y' - X_A \beta_A) + \text{Constant} \nonumber \\
&=& -\frac{1}{2{\sigma}^2} \left({y'}^T y'- 2{\beta_A}^T X_A^T y' \right)  + \text{Constant}. \nonumber
\end{eqnarray}

$\widehat{\beta_A}$ is the maximum likelihood estimator, that is, in a linear model, by least squares:
$$\widehat{\beta_A} (y') = (X_A^T X_A)^{-1} X_A^T y' \ \ \ \text{thus} \ \ \ X_A^T y' = X_A^T X_A \widehat{\beta_A} (y').$$
Since this is a linear model, the total sum of squares of $y'$ decomposes into the sum of squares explained by $X_A$ plus the sum of squared residuals:
\begin{eqnarray}
\sum_{i=1}^n {{y_i}'}^2 &=& \sum_{i=1}^n \left( \sum_{j \in A} \widehat{\beta_A} (y')_j X_{ij} \right)^2 + \sum_{i=1}^n \left( {y_i}' - \sum_{j \in A} \widehat{\beta_A} (y')_j X_{ij} \right)^2 \nonumber \\
{y'}^T y' &=& \left( X_A \widehat{\beta_A} (y') \right)^T X_A \widehat{\beta_A} (y') + n \widehat{\sigma^2} (y') \nonumber \\
\log p_Y (y') &=& -\frac{1}{2{\sigma}^2} \left( \left(X_A \widehat{\beta_A} (y') \right)^T X_A \widehat{\beta_A} (y') + n \widehat{\sigma^2} (y') - 2 {\beta_A}^T X_A^T X_A \widehat{\beta_A} (y') \right) \nonumber \\
&+& \text{Constant}. \nonumber
\end{eqnarray}
$p_Y (y')$ is therefore a deterministic function of $\widehat{\theta_A} (y') = (\widehat{\beta_A} (y'), \widehat{\sigma} (y'))$. By construction, $\widehat{\theta_A}$ is constant on $S \left( \widehat{\theta_A} (y) \right)$ so the density $p_Y$ is constant on the set. Therefore, the distribution of $Y$ conditional on $Y \in S \left( \widehat{\theta_A} (y) \right)$ is uniform.
\end{proof}

\begin{lemma} \label{discrete uniform}
Under $H_0 (A)$ and under the continuity assumption, the distribution of $\widehat{p_A} (Y)$ is the discrete uniform distribution over $\{0, 1/N, .., 1-1/N, 1\}$.
\end{lemma}
\begin{proof}
Let $k \in \{0, .., N\}$. The distribution of $\widehat{p_A} (Y)$ conditional on $p_A (Y)$ is given by $N \widehat{p_A} (Y) \mid p_A (Y) \sim \mathrm{Bin} (N,p_A (Y))$, that is:
$$\mathrm{P} \left(\widehat{p_A} (Y) = \frac{k}{N} \mid p_A (Y) \right) = {N \choose k} {p_A (Y)}^k \left( 1-p_A (Y) \right)^{N-k}$$

Corollary \ref{uniform distribution pAY} guarantees that under $H_0 (A)$ and under the continuity assumption, the distribution of $p_A (Y)$ is uniform over $[0,1]$. Therefore:
\begin{eqnarray}
\mathrm{P}_{H_0 (A)} \left(\widehat{p_A} (y) = \frac{k}{N} \right) &=&  \int_0^1 \mathrm{P} \left(\widehat{p_A} (y) = \frac{k}{N} \mid p_A (Y) = t \right) dt \nonumber \\
&=& {N \choose k} \int_0^1 t^k (1-t)^{N-k} dt \nonumber \\
&=& {N \choose k}  \frac{\Gamma(k+1) \Gamma(N-k+1)}{\Gamma(N+2)} \nonumber \\
&=& \frac{1}{N+1}. \nonumber
\end{eqnarray}
This probability does not depend on $k$, so the distribution of $\widehat{p_A} (y)$ is uniform under $H_0 (A)$ and the continuity assumption.
\end{proof}

To prove the analogous result which does not need the continuity assumption (lemma \ref{discrete stochastic dominance}), we we use the following result on stochastic dominance (see lecture notes by Olivier Bos, \textit{Stochastic Dominance: Theory and Applications}):

\begin{lemma} \label{dominance criterion}
A distribution $P_1$ stochastically dominates another distribution $P_2$ if and only if there exist two random variables $U_1$ and $U_2$ such that $U_1 \sim P_1$, $U_2 \sim P_2$, and $U_1 \geq U_2$ almost surely.
\end{lemma}
Without detailing the proof, note that the direction from the existence of such variables to stochastic dominance is easily verified, and in the reverse direction, if $P_2$ is the uniform distribution on $[0,1]$ (which is our case of interest here), $U_2$ can be constructed from $U_1$ by the conditional distribution $U_2 \mid (U_1 = u) \sim \text{Unif} \left(\left[\mathrm{P}(U_1 < u), \mathrm{P}(U_1 \leq u) \right]\right)$.

Lemmas \ref{discrete uniform} and \ref{dominance criterion} allow us to prove the following result on the distribution of $\widehat{p_A} (Y)$, without the continuity assumption:

\begin{lemma}\label{discrete stochastic dominance}
Under $H_0 (A)$, the distribution of $\widehat{p_A} (y)$ stochastically dominates the discrete uniform distribution over $\{0, 1/N, .., 1-1/N, 1\}$.
\end{lemma}

\begin{proof}
According to lemma \ref{dominance criterion} and the stochastic dominance of the distribution of $p_A(Y)$ over the uniform distribution on $[0,1]$ (lemma \ref{inequality pAY}), there exists a random variable $p_A^u$ uniformly distributed over $[0,1]$ such that $p_A(Y) \geq p_A^u$ almost surely (a.s.).

Let ${U_1, .., U_N}$ be independent random variables each uniformly distributed over $[0,1]$. We define the random variables:
\begin{eqnarray}
\hat{p}_1 &=& \frac{1}{N} \sum_{t=1}^N \mathbf{1} \left\{ U_t \leq p_A(Y) \right\} \nonumber \\
\hat{p}_2 &=& \frac{1}{N} \sum_{t=1}^N \mathbf{1} \left\{ U_t \leq p_A^u \right\} \nonumber
\end{eqnarray}
$p_A(Y) \geq p_A^u$ a.s., so for all $t$, $\left\{ U_t \leq p_A(Y) \right\} \geq \left\{ U_t \leq p_A^u \right\}$ a.s., thus $\hat{p}_1 \geq \hat{p}_2$ a.s. Therefore, the distribution of $\hat{p}_1$ stochastically dominates that of $\hat{p}_2$.

Moreover, just like $\widehat{p_A} (y)$, conditional on $p_A(Y)$ and $p_A^u(Y)$, $\hat{p}_1$ and $\hat{p}_2$ are empirical means of independent identically distributed Bernoulli random variables. Thus, we have the conditional distributions:
\begin{eqnarray}
N \hat{p}_1 \mid p_A(Y) &\sim& \mathrm{Bin} (N, p_A (Y)) \nonumber \\
N \hat{p}_2 \mid p_A^u &\sim& \mathrm{Bin} (N, p_A^u) \nonumber
\end{eqnarray}
$\hat{p}_1$ and $\widehat{p_A} (y)$ follow the same distribution conditional on $p_A(Y)$, so they follow the same distribution overall. $\hat{p}_2$ follows the same distribution as $p_A(Y)$ in lemma \ref{discrete uniform}: a binomial distribution with a random parameter uniformly distributed over $[0,1]$ (divided by $N$), and this lemma indicates that it is the discrete uniform distribution over $\{0, 1/N, .., 1-1/N, 1\}$.

Since the distribution of $\hat{p}_1$ stochastically dominates that of $\hat{p}_2$, the distribution of $\widehat{p_A} (y)$ stochastically dominates the discrete uniform distribution over $\{0, 1/N, .., 1-1/N, 1\}$.
\end{proof}

\begin{lemma} \label{single parameter calibration}
In single-parameter models, if $\widehat{\beta_A} (Y^{(1)}) = \beta_{A}^{(1)}$ that is, $\hat{e}^{(1)} = e^{(1)}$, then:
\begin{itemize}
\item $\mathrm{Var} \left[ \hat{e}^{(2)} | Y^{(1)} \right] = e^{(2)}$
\item $\mathrm{Var} \left( \hat{e}^{(2)} | Y^{(1)} \right) \leq \frac{1}{n} \left| e^{(2)} - e^{(1)} \right|$ with equality in the binary model.
\end{itemize}
\end{lemma}

\begin{proof}
In single-parameter models, estimating this parameter by maximum likelihood reduces to computing the empirical mean: $\hat{e}^{(m)} = \overline{Y^{(m)}}$.

Assuming that $\widehat{\beta_A} (Y^{(1)}) = \beta_{A}^{(1)}$, we have $\overline{Y^{(1)}} = e^{(1)}$. In both the binary model and the Poisson model,

$$\mathrm{E} \left[ \hat{e}^{(2)} | Y^{(1)} \right] = \mathrm{E} \left[ \overline{Y^{(2)}} | Y^{(1)} \right] = \frac{1}{n} \sum_{i=1}^{n} \mathrm{E} \left[Y^{(2)}_i | Y^{(1)} \right] = \overline{Z}.$$

$Z$ is defined linearly so either $\overline{Z} = \frac{e^{(2)}}{e^{(1)}} \overline{Y^{(1)}}$ (binary model with $e^{(2)} \leq e^{(1)}$, and Poisson model), or $\overline{Z} = 1 - \frac{1 - e^{(2)}}{1 - e^{(1)}} (1 - \overline{Y^{(1)}})$ (binary model with $e^{(2)} \geq e^{(1)}$). In both cases, since $\overline{Y^{(1)}} = e^{(1)}$, indeed $\overline{Z} = e^{(2)}$.

In both the binary and the Poisson model,
$$\mathrm{Var} \left( \hat{e}^{(2)} | Y^{(1)} \right) = \mathrm{Var} \left( \overline{Y^{(2)}} | Y^{(1)} \right) = \frac{1}{n^2} \sum_{i=1}^{n} \mathrm{Var} \left( Y^{(2)}_i | Y^{(1)} \right).$$
For any $p \in [0,1]$, the variance of a $\mathrm{Bernoulli} (p)$ is $p (1-p)$. Therefore, in the binary model,
$$\mathrm{Var} \left( \hat{e}^{(2)} | Y^{(1)} \right) = \frac{1}{n^2} \sum_{i=1}^{n} Z_i (1 - Z_i).$$
If $e^{(2)} \leq e^{(1)}$,
\begin{eqnarray}
\forall i, Z_i (1 - Z_i) &=& \frac{e^{(2)}}{e^{(1)}} Y_i \left(1 - \frac{e^{(2)}}{e^{(1)}} Y_i \right) \nonumber \\
 &=& 0 \ \text{if} \ Y_i = 0, \ 1 - \frac{e^{(2)}}{e^{(1)}} \ \text{if} \ Y_i = 1 \nonumber \\
\mathrm{Var} \left( \hat{e}^{(2)} | Y^{(1)} \right) &=& \frac{1}{n^2} \sum_{Y_i = 1} \left(1 - \frac{e^{(2)}}{e^{(1)}} \right) \nonumber \\
&=& \frac{1}{n^2} (n e^{(1)}) \left(1 - \frac{e^{(2)}}{e^{(1)}} \right) \ \text{since} \ \overline{Y^{(1)}} = e^{(1)} \nonumber \\
\mathrm{Var} \left( \hat{e}^{(2)} | Y^{(1)} \right) &=& \frac{1}{n} \left(e^{(1)} - e^{(2)} \right). \nonumber
\end{eqnarray}
And if $e^{(2)} \geq e^{(1)}$,
\begin{eqnarray}
\forall i, Z_i (1 - Z_i) &=& \frac{1 - e^{(2)}}{1 - e^{(1)}} (1 - Y_i) \left(1 - \frac{1 - e^{(2)}}{1 - e^{(1)}} (1 - Y_i) \right) \nonumber \\
 &=& 0 \ \text{if} \ Y_i = 1, \ 1 - \frac{1 - e^{(2)}}{1 - e^{(1)}} \ \text{if} \ Y_i = 0 \nonumber \\
\mathrm{Var} \left( \hat{e}^{(2)} | Y^{(1)} \right) &=& \frac{1}{n^2} \sum_{Y_i = 0} \left(1 - \frac{1 - e^{(2)}}{1 - e^{(1)}} \right) \nonumber \\
&=& \frac{1}{n^2} n (1 - e^{(1)}) \left(1 - \frac{1 - e^{(2)}}{1 - e^{(1)}} \right) \ \text{because} \ \overline{1 - Y^{(1)}} = 1 - e^{(1)} \nonumber \\
\mathrm{Var} \left( \hat{e}^{(2)} | Y^{(1)} \right) &=& \frac{1}{n} \left(e^{(2)} - e^{(1)} \right). \nonumber
\end{eqnarray}
In both cases $\mathrm{Var} \left( \hat{e}^{(2)} | Y^{(1)} \right) = \frac{1}{n} \left| e^{(2)} - e^{(1)} \right|$.

In the Poisson model, denote $\forall i, \ T_i = Z_i - \left\lfloor Z_i \right\rfloor$. Then:
$$T_i (1 - T_i) \leq \mathrm{min} (T_i, 1 - T_i) = \mathrm{min} \left( \left| Z_i - \left\lfloor Z_i \right\rfloor \right|, \left| Z_i - \left\lceil Z_i \right\rceil \right| \right).$$
The integer $k$ that minimizes $|Z_i - k|$ is either $\left\lfloor Z_i \right\rfloor$ or $\left\lceil Z_i \right\rceil$. Thus, since $Y^{(1)}_i$ is an integer,
$$T_i (1 - T_i) \leq |Z_i - Y^{(1)}_i| = \left|\frac{e^{(2)}}{e^{(1)}} - 1\right| Y^{(1)}_i.$$
\begin{eqnarray}
\mathrm{Var} \left( \hat{e}^{(2)} | Y^{(1)} \right) &=& \frac{1}{n^2} \sum_{i = 1}^n T_i (1 - T_i) \leq \frac{1}{n^2} \sum_{i = 1}^n \left|\frac{e^{(2)}}{e^{(1)}} - 1\right| Y^{(1)}_i \nonumber \\
\mathrm{Var} \left( \hat{e}^{(2)} | Y^{(1)} \right) &\leq& \frac{1}{n^2} \left|\frac{e^{(2)}}{e^{(1)}} - 1 \right| (n e^{(1)}) = \frac{1}{n} \left| e^{(2)} - e^{(1)} \right|. \nonumber
\end{eqnarray}

\end{proof}

\section{Extended theorem: control of selection error} \label{extended theorem}
\subsection{Statement of the theorem}
\begin{theorem}
Assume that the set of covariates is divided into three disjoint subsets: $\{1, .. ,p \} = A \cup B \cup C$ with $\forall j \in C, \beta_j = 0$, and that the active variables are orthogonal to the inactive variables, that is:
$$X_C^T X_{A\cup B}=0.$$
Then, the simulation-calibration test of $H_0 (A)$ at level $\alpha$ has a probability less than $\alpha + \frac{1 - \alpha}{N+1}$ of selecting a false positive, that is, a rejection of the test in a case where $j_A (Y) \in C$.
\end{theorem}

\subsection{Proof of the theorem}
\begin{proof}

Let $\widehat{\beta}^{Lasso-A \cup C} (\lambda,Y)$ be the Lasso estimate at the parameter $\lambda$ for the linear model restricted to $A\cup CB$, and let:
$$\lambda_{A-C} (Y) = \mathrm{sup} \left\{ \lambda \geq 0 : \exists j \in C, \widehat{\beta_j}^{Lasso-A \cup C} (\lambda,Y) \neq 0 \right\}.$$
Define $\beta ^{Lasso - A \cup B} (\lambda,Y)$ and $\lambda_{A-B} (Y)$ similarly by replacing $C$ with $B$, and denote $\beta^{Lasso - A} (\lambda,Y)$ as the Lasso estimate of the model restricted to $A$. Since the Lasso paths of $\widehat{\beta}^{Lasso-A\cup B}$, $\widehat{\beta}^{Lasso-A\cup C}$, and $\widehat{\beta}^{Lasso}$ coincide as long as only variables belonging to $A$ have been selected, we have $\lambda_A (Y) = \max \left(\lambda_{A-B} (Y),\lambda_{A-C} (Y) \right)$. The test produces a false positive when both $\widehat{p_A} (Y) \leq \alpha$ ($H_0 (A)$ is rejected and a variable is selected) and $\lambda_{A-C} (Y) \geq \lambda_{A-B} (Y)$ (the selected variable is a false positive). $\widehat{p_A} (Y)$ takes small values when $\lambda_A (Y)$ (equal to $\lambda_{A-C} (Y)$ in the case of a false positive) takes large values. Therefore, we look for a dominance over the distribution of $\lambda_{A-C} (Y)$, which will translate into a dominance by the distribution of $\widehat{p_A} (Y)$.

Let $\lambda > 0$. By definition of $\lambda_{A-C}$,
$$\lambda_{A-C} (Y) > \lambda \Leftrightarrow \exists \lambda' > \lambda :  \exists j \in C :  \beta_j^{Lasso-A \cup C} (\lambda',Y) \neq 0$$
At every point along the Lasso path, the covariance of the residuals with each active variable is equal to $\lambda'$. Thus:
$$\lambda_{A-C} (Y) > \lambda  \Leftrightarrow  \exists \lambda' > \lambda:  \exists j in C:  \left| \mathrm{Cov} \left(Y-X_{A\cup C} \widehat{\beta}^{Lasso-A\cup C} (\lambda',Y),X_j \right) \right| = \lambda'$$
Furthermore, as long as no variables outside of $A$ are selected, $\widehat{\beta}^{Lasso-A\cup C}$ and $\widehat{\beta}^{Lasso-A}$ coincide. So:
$$\lambda_{A-C} (Y) > \lambda  \Leftrightarrow  \exists \lambda' > \lambda:  \exists j in C:  \left| \mathrm{Cov} \left( Y - X_A \widehat{\beta}^{Lasso-A} (\lambda',Y),X_j \right) \right|= \lambda'.$$
Since $\forall j \in C,X_j^T X_A = 0$, we have:
\begin{eqnarray} \label{covariance inequality}
\lambda_{A-C} (Y) > \lambda  &\Leftrightarrow&  \exists \lambda' > \lambda:  \exists j in C:  \left| \mathrm{Cov} (Y,X_j ) \right| = \lambda' \nonumber \\
\lambda_{A-C} (Y) > \lambda  &\Leftrightarrow&  \exists j in C: \left| \mathrm{Cov} ( Y,X_j ) \right| > \lambda.
\end{eqnarray}
Thus $\lambda_{A-C} (Y) = \underset{j \in C}{\max} \left| \mathrm{Cov} ( Y,X_j ) \right|$.

Let $Y_0$ and $Y_1$ be two random vectors satisfying the model \ref{linmodel}, with the same residual variance $\sigma^2$ and coefficient vectors $\beta_0$ and $\beta_1$ which differ only by their coefficients on $B$: ${\beta_{0}}_{A} = {\beta_{1}}_{A}$ is arbitrary, while ${\beta_{0}}_{B} = 0$, ${\beta_{1}}_{B}$ is arbitrary, and ${\beta_{0}}_{C} = {\beta_{1}}_{C} = 0$. In other words, $Y_0$ satisfies $H_0 (A)$ and $Y_1$ does not necessarily satisfy it. We want to prove the theorem on the arbitrary vector $Y_1$, based on comparisons with the known case of the null hypothesis ($Y_0$).

Also let $\theta_{A}' = (\beta_{A}',\sigma_{A}' ) \in \Theta_A$ be a parameter vector of the model restricted to $A$. In the following four subsections, we prove that the distribution of $\lambda_{A-C} (Y_1)$ conditional on $\widehat{\theta_A} (Y_1) = \theta_{A}'$ is stochastically dominated by the distribution of $\lambda_{A-C} (Y_0)$ conditional on $\widehat{\theta_A} (Y_0) = \theta_{A}'$.

\subsubsection{Decomposition of the vector $Y$}
Consider the linear regression of each of the variables belonging to $B$ on those belonging to $A$:
$$X_B = X_A \Pi_{AB} + \widetilde{X}_B, \ \ \widetilde{X}_B^T X_A = 0.$$

This allows us to define, in $Y = Y_0$ or $Y = Y_1$, the component of the expectation of $Y$ due exclusively to the variables of $B$: 
$$y_B = \widetilde{X}_B \beta_B.$$

By assumption, ${y_B}_0 = 0$, but ${y_B}_1$ may be non-zero. We represent it as ${y_B}_1 = ||{y_B}_1|| {u_B}_1$ where ${u_B}_1$ is a unit vector (if ${y_B}_1 = 0$, then any unit vector in $\mathrm{Vect}\left(\widetilde{X}_B \right)$ is chosen for ${u_B}_1$).

The model \ref{linmodel} can then be written as:

$$Y = X_A \left(\beta_A + \Pi_{AB} \beta_B \right) + y_B + \epsilon$$

or, since ${y_B}_0 = 0$, the following representation is valid both for $h = 0$ and $h = 1$, denoting $\gamma_h = ||{y_B}_h||$:
\begin{equation} \label{y_B model}
Y_h = X_A \left(\beta_A + \Pi_{AB} {\beta_h}_B \right) + \gamma_h {u_B}_1 + \epsilon_h, \ \ \ \epsilon_h \sim \mathcal{N} \left(0,\sigma^2 I_n \right).
\end{equation}

We consider the linear regression associated with this expression, namely the regression of $Y$ on the variables $X_A$ and the vector ${u_B}_1$:
\begin{equation} \label{u_B regression}
Y = X_A \widehat{\beta}_A + \widehat{\gamma} {u_B}_1 + r, \ \ r^T X_A = 0, \ r^T {u_B}_1 = 0.
\end{equation}

Due to the orthogonality between the variables in this regression (${u_B}_1$ is a linear combination of $\widetilde{X}_B$ thus ${u_B}_1^T X_A = 0$), the $\widehat{\beta}_A$ obtained is the same as in the linear regression of $Y$ on $X_A$ only.

$\widehat{\beta}_A$, $\widehat{\gamma}$, and $r$ are all functions of $Y$, thus are random variables. We will characterize their distributions, distinguishing as necessary between $Y = Y_0$ and $Y = Y_1$, and then their distributions conditional on $\widehat{\theta_A} (Y) = \theta_{A}'$.

\subsubsection{Non-conditional distributions of the regression components}
$X_A \widehat{\beta}_A$, $\widehat{\gamma} {u_B}_1$, and $r$ are the orthogonal projections of $Y$ onto the mutually orthogonal spaces of, respectively, $\mathrm{Vect} (X_A)$, $\mathrm{Vect} ({u_B}_1)$, and the orthogonal complement of these two spaces, denoted $V_r$. Since $Y$ is a Gaussian vector, all three are Gaussian vectors, and they are centered at the orthogonal projections of $\mathrm{E}[Y] = X_A \left(\beta_A + \Pi_{AB} {\beta_h}_B \right) + \gamma_h {u_B}_1$ onto these three spaces (in other words, the linear regression is unbiased):
\begin{eqnarray}
E \left[X_A \widehat{\beta}_A \right] &=& \beta_A + \Pi_{AB} {\beta_h}_B \nonumber \\
E \left[\widehat{\gamma} {u_B}_1 \right] &=& \gamma_h {u_B}_1 \nonumber \\
E[r] &=& 0. \nonumber
\end{eqnarray}

Moreover, the covariance matrix of $Y$, which is $\sigma^2 I_n$, remains $\sigma^2 I_n$ when considered in an orthogonal basis spanned by these three subspaces. Therefore, the three Gaussian vectors $X_A \widehat{\beta}_A$, $\widehat{\gamma} {u_B}_1$, and $r$ are independent (hence the variables $\widehat{\beta}_A$, $\widehat{\gamma}$, and $r$ are independent too), and their projected variance over any dimension of their respective support subspace is always $\sigma^2$. In particular,
\begin{itemize}
    \item[*] $r$ follows the centered multivariate normal distribution with covariance matrix $I_{\mathrm{dim}(V_r)} = I_{p - |A| - 1}$, i.e., the density on $V_r$ is:
    $$f_r(u) =  \frac{1}{{\left(\sigma \sqrt{2 \pi} \right)}^{p - |A| - 1}} e^{- \frac{{||u||}^2}{2 \sigma^2}}.$$
    This distribution is the same whether considering $r(Y_0)$ or $r(Y_1)$.
    \item[*] $\widehat{\gamma} \sim \mathcal{N} \left(\gamma_h,\sigma^2 \right)$. The expectation thus varies depending on whether we consider $Y_0$ or $Y_1$: $\widehat{\gamma} (Y_0) \sim \mathcal{N} \left(0,\sigma^2 \right)$ and $\widehat{\gamma} (Y_1) \sim \mathcal{N} \left(\gamma_1,\sigma^2 \right)$.
\end{itemize}
This result implies a strong form of dominance of the distribution of $| \widehat{\gamma} (Y_1) |$ over that of $| \widehat{\gamma} (Y_0) |$:

\begin{definition}[monotone likelihood ratio property] Let $P_0$ and $P_1$ be two probability distributions on $\mathbb{R}$ that admit densities $f_0$ and $f_1$ with the same support $S$. We say that $P_1$ satisfies the monotone likelihood ratio property over $P_0$ if $\frac{f_1}{f_0}$ is increasing on $S$.
\end{definition}

\begin{lemma} \label{RVM gamma}
 The distribution of $| \widehat{\gamma} (Y_1) |$ satisfies the monotone likelihood ratio property over that of $| \widehat{\gamma} (Y_0) |$.
\end{lemma}
\begin{proof}
For $h = 0$ or $h = 1$, for all $t \in \mathbb{R}_{+}$,
$$f_{|\widehat{\gamma}|} (t) = f_{\widehat{\gamma}} (t) + f_{\widehat{\gamma}} (-t) $$
thus:
\begin{eqnarray}
\frac{f_{| \widehat{\gamma} (Y_1) |}}{f_{| \widehat{\gamma} (Y_0) |}} (t) &=& \frac{ e^{- \frac{(t - \gamma_1)^2}{2 \sigma^2}} + \ e^{- \frac{(t + \gamma_1)^2}{2 \sigma^2}}}{2 e^{- \frac{t^2}{2 \sigma^2}}} \nonumber \\
&=& \frac{1}{2} \left( e^{- \frac{(t - \gamma_1)^2 - t^2}{2 \sigma^2}} + \ e^{- \frac{(t + \gamma_1)^2 -t^2}{2 \sigma^2}} \right) \nonumber \\
&=& \frac{1}{2} e^{- \frac{\gamma_1^2}{2 \sigma^2}} \mathrm{cosh} \left( \frac{\gamma_1 t}{\sigma^2} \right) \nonumber
\end{eqnarray}
Since the hyperbolic cosine function is increasing on $\mathbb{R}_{+}$, the monotone likelihood ratio property is verified.
\end{proof}

The monotone likelihood ratio property implies the stochastic dominance of $P_1$ over $P_0$ \citep{roosen_testing_2004}. Furthermore, it is preserved by applying the same derivable strictly increasing function $h$ to random variables, as this amounts to introducing a multiplicative term that cancels out in the ratio of the densities:
$$\frac{f_{h(U_1)}}{f_{h(U_0)}} (h(t)) = \frac{h'(t) f_{U_1} (t)}{h'(t) f_{U_0} (t)} = \frac{f_{U_1}}{f_{U_0}} (t).$$
Therefore, by applying the square function, the distribution of ${\widehat{\gamma} (Y_1)}^2$ satisfies the monotone likelihood ratio property over that of ${\widehat{\gamma} (Y_0)}^2$.

\subsubsection{Conditional distribution of the residual vector}

Having obtained these results on the non-conditional distributions of $\widehat{\beta}_A (Y)$, $\widehat{\gamma} (Y)$, and $r (Y)$, we can now derive results on the distribution of the third of these components, the residual vector $r (Y)$, conditional on $\widehat{\theta_A} (Y) = \theta_{A}'$. The residual from the regression of $Y$ on $X_A$ alone is equal to $\widehat{\gamma} (Y) {u_B}_1 + r (Y)$, hence:
$$ n \widehat{\sigma_A}^2 = {|| \widehat{\gamma} {u_B}_1 + r (Y) ||}^2 = {\widehat{\gamma} (Y)}^2 + {|| r (Y) ||}^2.$$
The condition $\widehat{\theta_A} (Y) = \theta_{A}'$ is therefore equivalent to:
\begin{equation}\label{condition thetaprime}
\widehat{\beta}_A (Y) = \beta_{A}' \ \ , \ \ {\widehat{\gamma} (Y)}^2 + {|| r (Y) ||}^2 = n {\sigma_{A}'}^2.
\end{equation}

Since $\widehat{\beta}_A (Y)$, $\widehat{\gamma} (Y)$, and $r (Y)$ are independent, only the second equality affects the distributions of $\widehat{\gamma} (Y)$ and $r (Y)$. This will allow us to transition from the stochastic dominance between the $\widehat{\gamma} (Y)$ obtained in the previous section to a form of stochastic dominance between the $r (Y)$. For this, we need a result on preserving dominance when transitioning from non-conditional to conditional distribution. This result is about the stronger form of dominance, the monotone likelihood ratio, and not stochastic dominance. The equivalent statement about stochastic dominance is not true.

\begin{lemma} \label{Conditional MLR}
Let $U_0$ and $U_1$ be two real-valued random variables with the same support such that the distribution of $U_1$ satisfies the monotone likelihood ratio property over that of $U_0$, and let $V$ be a real-valued random variable that is independent of both $U_0$ and $U_1$ and follows a density distribution.

Then, the distributions of $U_0 + V$ and $U_1 + V$ share the same support and for all $a$ in that support, the conditional distribution of $U_1$ given $U_1 + V = a$ satisfies the monotone likelihood ratio property over the conditional distribution of $U_0$ given $U_0 + V = a$.
\end{lemma}

\begin{proof}
Let $f_0$ and $f_1$ be the probability density functions of the distributions of $U_0$ and $U_1$, respectively, and $S$ their common support. Let $g$ be the probability density function of the distribution of $V$. The support of $U_0 + V$ and $U_1 + V$ is $S' = \{a : \int_{S} f_h (u) g(a - u) du > 0 \}$.
\\
Let $a \in S'$. For $h = 0$ or $1$, let $\tilde{f}_{a, h}$ be the density of the distribution of $U_h$ given $U_h + V = a$. It is defined as:

$$\forall t \in \mathbb{R}, \tilde{f}_{a, h} (t) = \frac{f_h (t) g(a - t)}{\int_{S} f_h (u) g(a - u) du}.$$

Thus, the conditional distributions of $U_0$ given $U_0 + V = a$ and of $U_1$ given $U_1 + V = a$ have the same support, $S''(a) = \{t \in S : g(a - t) > 0 \}$, and:

$$\forall t \in S''(a), \ \frac{\tilde{f}_{a, 1}}{\tilde{f}_{a, 0}} (t) = \frac{f_1 (t)}{f_0 (t)} \frac{\int_{S} f_0 (u) g(a - u) du}{\int_{S} f_1 (u) g(a - u) du}.$$

With $a$ fixed, the ratio of the two integrals is a constant. The ratio $f_1 / f_0$ is increasing due to the monotone likelihood ratio property of the distribution of $U_1$ over that of $U_0$. Therefore, this property holds for the conditional distributions.
\end{proof}

Applying this lemma to the second part of the conditioning \ref{condition thetaprime}, since the distribution of ${\widehat{\gamma} (Y_1)}^2$ verifies the monotone likelihood ratio property over that of ${\widehat{\gamma} (Y_0)}^2$ (a consequence of Lemma \ref{RVM gamma}), the conditional distribution of ${\widehat{\gamma} (Y_1)}^2$ given $\widehat{\theta_A} (Y_1) = \theta_{A}'$ also verifies this property over the distribution of ${\widehat{\gamma} (Y_0)}^2$ given $\widehat{\theta_A} (Y_0) = \theta_{A}'$. Therefore, it stochastically dominates it.

We can now prove the following result, which is a strong and multivariate form of stochastic dominance of the conditional distribution of $r (Y_0)$ over that of $r (Y_1)$.

\begin{lemma} \label{multivariate stochastic dominance}
There exist two random vectors on $V_r$, $R_0$ and $R_1$, such that:
\begin{itemize}
    \item For $h = 0$ or $h = 1$, the distribution of $R_h$ is identical to the distribution of $r (Y_h)$ conditional on $\widehat{\theta_A} (Y_h) = \theta_{A}'$.
    \item Almost surely (a.s.), $R_0$ and $R_1$ are positively collinear and $|| R_0 || \geq || R_1 ||$.
\end{itemize}
\end{lemma}
This characterization is a multivariate analogue of the criterion for dominance between univariate variables introduced by Lemma \ref{dominance criterion}.

\begin{proof}
As the distribution of ${\widehat{\gamma} (Y_0)}^2$ conditional on $\widehat{\theta_A} (Y_0) = \theta_{A}'$ is stochastically dominated by that of ${\widehat{\gamma} (Y_1)}^2$ conditional on $\widehat{\theta_A} (Y_1) = \theta_{A}'$, according to lemma \ref{dominance criterion}, there exist two random variables $G_0$ and $G_1$, with the same distributions as these two conditional distributions respectively, such that $G_0 \leq G_1$ a.s.

For all $d \geq 0$, let $S_{V_r, d^2} = \left\{u \in V_r : {||u||}^2 = d^2 \right\}$ be the centered sphere of radius $d$ in the residual space $V_r$. The probability density of $r(Y)$ is proportional to $u \rightarrow e^{-\frac{1}{2 {||u||}^2 }}$, so it is constant on each of the $S_{V_r, d^2}$.

Let $U$ be a random vector uniformly distributed on the unit sphere $S_{V_r, 1}$. Define $R_0$ and $R_1$ by:
$$\forall h \in \{0,1 \}, \ R_h = \sqrt{ n {\sigma_{A}'}^2 - G_h } \ U.$$
$R_0$ and $R_1$ are positively collinear and since $G_0 \leq G_1$ a.s., $|| R_0 || \geq || R_1 ||$ a.s. There remains to prove that $R_0$ and $R_1$ follow the desired distributions.

Let $\gamma' \in \mathbb{R}$. The double condition $\left( \widehat{\theta_A}, \widehat{\gamma}^2 \right) (Y) = \left( \theta_{A}', {\gamma'}^2 \right)$ is equivalent to:
$$\widehat{\beta}_A (Y) = \beta_{A}' \ \ , \ \ {\widehat{\gamma} (Y)}^2 = {\gamma'}^2 \ \ , \ \ {|| r (Y) ||}^2 = n {\sigma_{A}'}^2 - {\gamma'}^2$$
Since $\widehat{\beta}_A (Y)$, $\widehat{\gamma} (Y)$, and $r (Y)$ are independent, only the third equality affects the distribution of $r (Y)$. The distribution of $r (Y)$ under the double condition is therefore the distribution of $r(Y)$ conditional on $r (Y) \in S_{V_r, n {\sigma_{A}'}^2 - {\gamma'}^2}$, which, since the density of $r(Y)$ is constant on this set, is a uniform distribution:
$$r(Y) | \left[ \left( \widehat{\theta_A}, \widehat{\gamma}^2 \right) (Y) = \left( \theta_{A}', {\gamma'}^2 \right) \right] \sim \mathcal{U} \left( S_{V_r, n {\sigma_{A}'}^2 - {\gamma'}^2} \right).$$

Moreover, for $h = 0$ or $h = 1$, due to its definition, $R_h$ follows the conditional distribution:
$$R_h | \left[ \widehat{\gamma}^2 (Y) = {\gamma'}^2 \right] \sim \mathcal{U} \left( S_{V_r, n {\sigma_{A}'}^2 - {\gamma'}^2} \right)$$
and the distribution of $G_h$ is identical to the distribution of ${\widehat{\gamma} (Y_h)}^2$ conditional on $\widehat{\theta_A} (Y_h) = \theta_{A}'$. Therefore, the distribution of $R_h$ without conditioning and the distribution of $r(Y_h)$ under only the conditioning $\widehat{\theta_A} (Y_h) = \theta_{A}'$ are obtained in the same way, by integrating the distribution $\mathcal{U} \left( S_{V_r, n {\sigma_{A}'}^2 - {\gamma'}^2} \right)$ over ${\gamma'}^2$ according to the distribution of $G_h$. Thus, they are equal.
\end{proof}

\subsubsection{Conclusion on the conditional distribution of $\lambda_{A-C}$} \label{conditional dist lambdaAC}
We get back to the equivalence \eqref{covariance inequality}, which characterizes the selection of a false positive. We inject the decomposition of $Y$ into projected components on different subspaces \eqref{u_B regression}. The first two terms of this decomposition are vectors orthogonal to $X_j, j \in C$; thus, $\forall j \in C, \ \mathrm{Cov}( Y,X_j ) = \mathrm{Cov}( r(Y),X_j )$.
Then for all $\lambda > 0$,
\begin{eqnarray}
\lambda_{A-C} (Y) \geq \lambda  &\Leftrightarrow& \left( \underset{j \in C}{\max} \left| \mathrm{Cov}( r(Y),X_j ) \right| \geq \lambda \right) \nonumber \\
\mathrm{P} \left( \lambda_{A-C} (Y) \geq \lambda \mid \widehat{\theta_A} (Y) = \theta_{A}' \right) &=& \mathrm{P} \left( \underset{j \in C}{\max} \left| \mathrm{Cov}( r(Y),X_j ) \right| \geq \lambda \mid \widehat{\theta_A} (Y) = \theta_{A}' \right). \nonumber 
\end{eqnarray}

The distribution of $\lambda_{A-C} (Y)$ conditional on $\widehat{\theta_A} (Y)$ depends only on the distribution of $r (Y)$ under the same conditioning. This distribution is described by Lemma \ref{multivariate stochastic dominance}, which allows for distinctions depending on whether the null hypothesis is satisfied or not. For $h = 0$ or $1$,
$$\mathrm{P} \left( \lambda_{A-C} (Y_h) \geq \lambda \mid \widehat{\theta_A} (Y) = \theta_{A}' \right) = \mathrm{P} \left( \underset{j \in C}{\max} \left| \mathrm{Cov}( R_h,X_j ) \right| \geq \lambda \mid \widehat{\theta_A} (Y) = \theta_{A}' \right).$$
Since $R_0$ and $R_1$ are positively collinear with $|| R_0 || \geq || R_1 ||$ almost surely. Then:
\begin{eqnarray}
\forall j \in C, \ \left| \mathrm{Cov}( R_0,X_j ) \right| &\geq& \left| \mathrm{Cov}( R_1,X_j ) \right| \ \text{a.s.} \nonumber \\
\left( \underset{j in C}{\max} \left| \mathrm{Cov}( R_1,X_j ) \right| \geq \lambda \right)  &\implies& \left( \underset{j in C}{\max} \left| \mathrm{Cov}( R_0,X_j ) \right| \geq \lambda \right) \ \text{a.s.} \nonumber \\
\mathrm{P} \left( \lambda_{A-C} (Y_1) \geq \lambda \mid \widehat{\theta_A} (Y_1) = \theta_{A}' \right) &\geq& \mathrm{P} \left( \lambda_{A-C} (Y_0) \geq \lambda \mid \widehat{\theta_A} (Y_0) = \theta_{A}' \right). \nonumber
\end{eqnarray}
Given this is true for every $\lambda$, the distribution of $\lambda_{A-C} (Y_0)$ conditional on $\widehat{\theta_A} (Y_0) = \theta_{A}'$ stochastically dominates the distribution of $\lambda_{A-C} (Y_1)$ conditional on $\widehat{\theta_A} (Y_1) = \theta_{A}'$.

\subsubsection{Conclusion of the proof of the theorem}
The test produces a false positive if and only if $\lambda_{A-C} (Y) \geq \lambda_{A-B} (Y)$ and $\widehat{p_A} (Y) \leq \alpha$, with $\widehat{p_A} (Y)$ itself being an estimator of the conditional p-value $p_A (Y)$. We consolidate this double condition into a single one by defining the random variables $\tilde{\lambda}_{A-C} (Y)$, $\tilde{p_A} (Y)$ and $\tilde{p_A} (Y)$ as follows:
\begin{eqnarray}
\text{if} \ \lambda_{A-C} (Y) \geq \lambda_{A-B} (Y) &:& \tilde{\lambda}_{A-C} (Y) = \lambda_{A-C} (Y) = \lambda_A (Y) \nonumber \\
&:& \tilde{p}_A (Y) = p_A (Y) \nonumber \\
&:& \widehat{\tilde{p}_A} = \widehat{p_A} (Y) \nonumber \\
\text{otherwise} &:& \tilde{\lambda}_{A-C} (Y) = 0 \nonumber \\
&:& \tilde{p}_A (Y) = \widehat{\tilde{p}_A} (Y) = 1. \nonumber
\end{eqnarray}
Therefore, the test produces a false positive if and only if $\widehat{\tilde{p}_A} (Y) \leq \alpha$. We will establish stochastic dominances between the distributions of these tilded variables, with or without conditioning, in $Y = Y_0$ or $Y = Y_1$.

Under the condition $\widehat{\theta_A} (Y) = \theta_{A}'$, by reformulating definition \ref{pAy}, $p_A (Y)$ is obtained by applying a decreasing function to $\lambda_{A} (Y)$:
\begin{eqnarray}
p_A (Y) &=& \phi_{\theta_{A}'} \left( \lambda_{A} (Y) \right) \ \ \ \text{where:} \nonumber \\
\forall \lambda \geq 0, \ \phi_{\theta_{A}'} (\lambda) &=& \mathrm{P} \left(\lambda_A (Y_0) \geq \lambda \mid \widehat{\theta_A} (Y_0) = \theta_{A}' \right). \nonumber 
\end{eqnarray}
Furthermore, since $\phi_{\theta_{A}'} (0) = 1$, we have $\tilde{p}_A (Y) = \phi_{\theta_{A}'} \left( \tilde{\lambda}_{A-C} (Y) \right)$.

$\lambda_{A} (Y) \geq \lambda_{A-C} (Y) \geq \tilde{\lambda}_{A-C} (Y)$ thus the distribution of $\lambda_{A} (Y_0)$ conditional on $\widehat{\theta_A} (Y_0) = \theta_{A}'$ stochastically dominates the distribution of $\tilde{\lambda}_{A-C} (Y_0)$ under the same conditioning. This in turn dominates the distribution of $\tilde{\lambda}_{A-C} (Y_1)$ conditional on $\widehat{\theta_A} (Y_1) = \theta_{A}'$ (subsection \ref{conditional dist lambdaAC}).

Therefore, due to the decreasing nature of $\phi_{\theta_{A}'}$, the distribution of $\tilde{p}_A (Y_1)$ conditional on $\widehat{\theta_A} (Y_1) = \theta_{A}'$ stochastically dominates the distribution of $p_A (Y_0)$ conditional on $\widehat{\theta_A} (Y_0) = \theta_{A}'$. Moreover, by construction, this latter distribution itself dominates the uniform distribution on $[0,1]$ (equation \eqref{conditional inequality} from the proof of lemma \ref{inequality pAY}).

The conditional distribution of $\tilde{p}_A (Y_1)$ dominates the uniform distribution for any value of $\theta_{A}'$, therefore the non-conditional distribution of $\tilde{p}_A (Y_1)$ itself dominates the uniform distribution.

Finally, since $N \widehat{p_A} (Y) \mid p_A (Y) \sim \mathrm{Bin} (N,p_A (Y))$ where $N$ is the number of simulations performed (lemma 5 of the main article), and since the distribution $\mathrm{Bin} (N,1)$ is that of a constant a.s. equal to $N$, we also have $N \widehat{\tilde{p}_A} (Y) \mid p_A (Y) \sim \mathrm{Bin} (N,\tilde{p}_A (Y))$. Thus, as in lemma \ref{discrete stochastic dominance}, the stochastic dominance of the uniform distribution over $[0,1]$ by the distribution of $\tilde{p}_A (Y_1)$ leads to the stochastic dominance of the discrete uniform distribution over $\{0,1/N,..,1\}$ by the distribution of $\widehat{\tilde{p}_A} (Y_1)$.

Thus, as under the null hypothesis, the probability of $\widehat{\tilde{p}_A} (Y_1) \leq \alpha$ (i.e., the selection of a false positive) is bounded by $\alpha + \frac{1 - \alpha}{N + 1}$.

\begin{conjecture}[Extended control of the selection error]
Let $\{1, .. ,p \} = A \cup B \cup C$ with $\forall j \in C, \beta_j = 0$.
Assume that the residuals of the variables in $B$ on those in $A$ are orthogonal to the residuals of the variables in $C$ on those in $A$, that is, there exist $\Pi_{AB}$, $\widetilde{X}_B$, $\Pi_{AC}$, $\widetilde{X}_C$ satisfying:
\begin{eqnarray}
X_B = X_A \Pi_{AB} + \widetilde{X}_B, \ \ \widetilde{X}_B^T X_A = 0 \nonumber \\
X_C = X_A \Pi_{AC} + \widetilde{X}_C, \ \ \widetilde{X}_C^T X_A = 0 \nonumber \\
\widetilde{X}_B^T \widetilde{X}_C = 0. \nonumber
\end{eqnarray}
Then, for any $\alpha \in [0,\frac{1}{2}]$, the simulation-calibration test of $H_0 (A)$ at level $\alpha$ has a probability less than $\alpha + \frac{1 - \alpha}{N+1}$ of selecting a false positive.
\end{conjecture}

The idea of orthogonality between residuals on $A$ is that the correlation between two variables of indices $j_B \in B$ (unknown active) and $j_C \in C$ (inactive), which normally prevents controlling the probability of selecting the false positive $j_C$ instead of $j_B$, no longer plays this role if it relies entirely on their common correlations to variables in $A$. The disturbance (relative to the null hypothesis) caused by the unknown active variable $j_B$ would be entirely absorbed by the calibration on $A$.

\end{proof}

\section{Supplementary figures}
\newpage
\begin{figure}
    \centering
    \includegraphics[width=\textwidth]{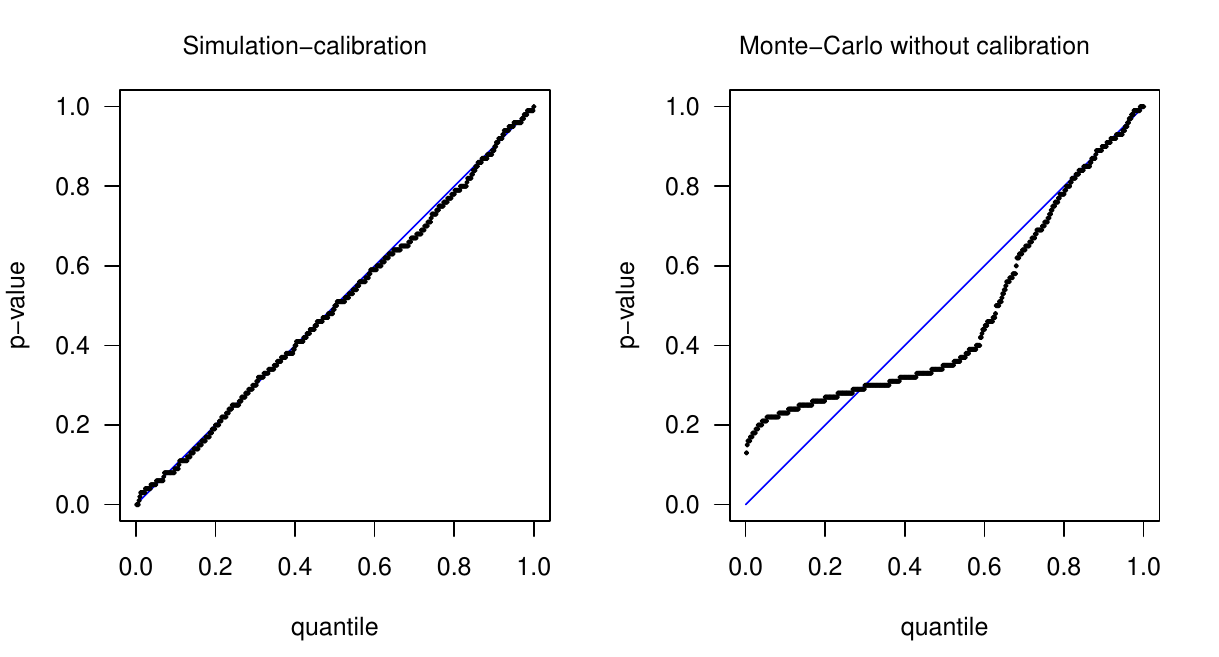}
    \caption{Impact of calibration in a scenario: linear model, $\rho = 0.99$, $1$ active variable, $\mathrm{SNR} = 0.1$.}
    \label{fig:example noncalibrated qqplot}
\end{figure}

\begin{figure}
    \centering
    \makebox[\textwidth][c]{\includegraphics[width=\textwidth]{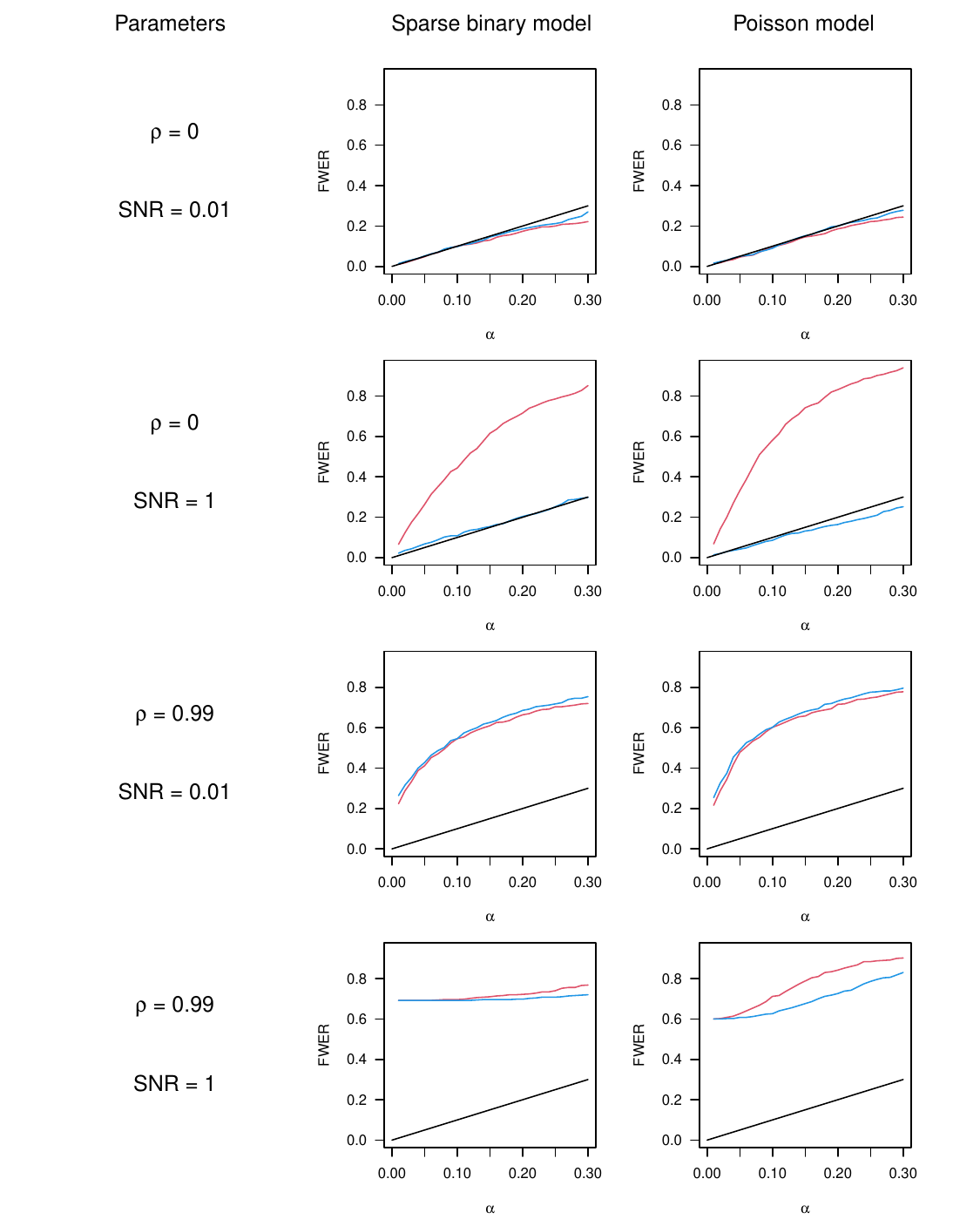}}
    \caption{FWER of the variable selection procedure with thresholding (blue) or \textit{ForwardStop} (red) as a function of $\alpha$ in 8 scenarios of sparse binary or Poisson model with 10 active variables.}
    \label{fig:procedure-simcal-FWER-2}
\end{figure}

\begin{figure}
    \centering
    \makebox[\textwidth][c]{\includegraphics[width=\textwidth]{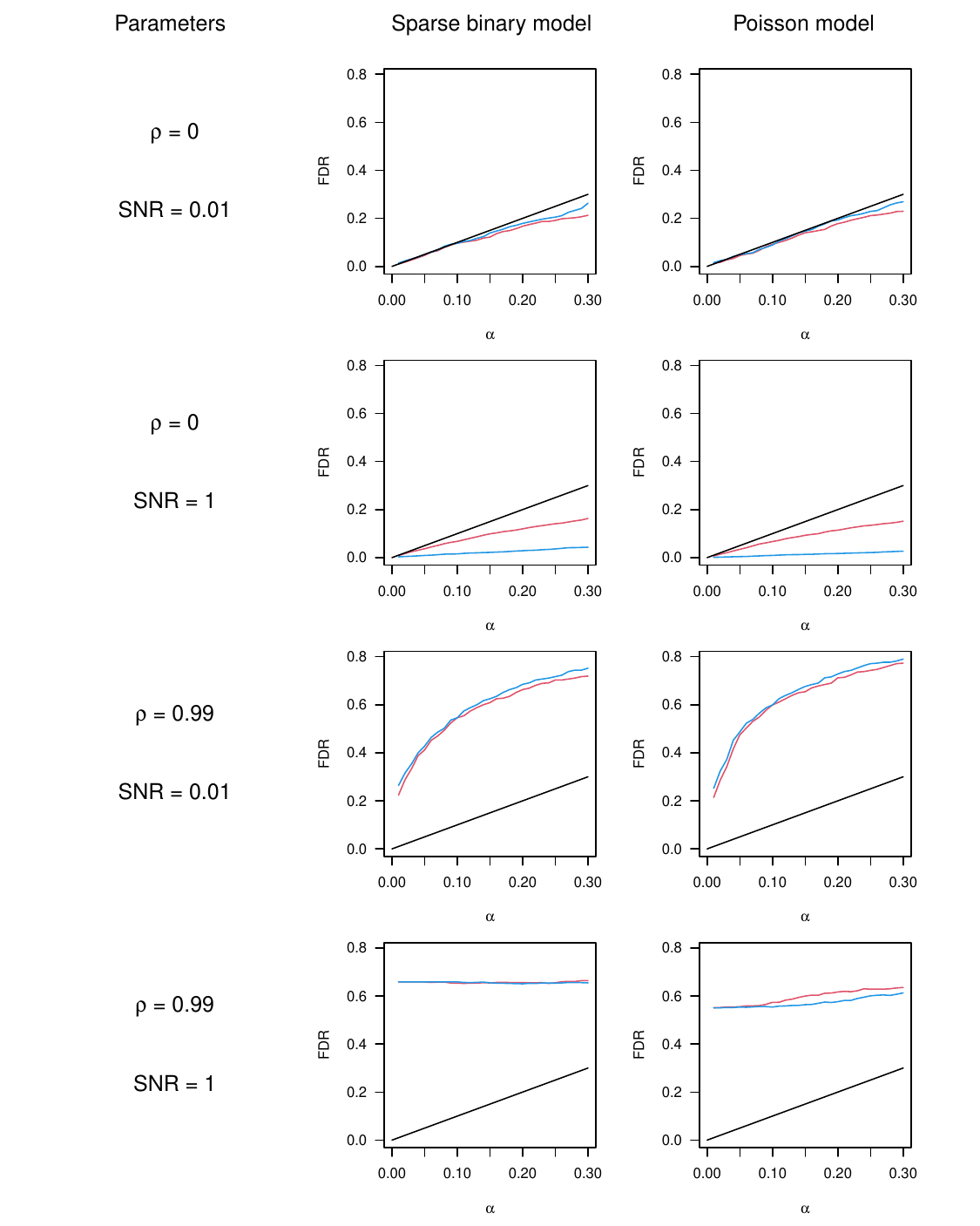}}
    \caption{FDR of the variable selection procedure with thresholding (blue) or \textit{ForwardStop} (red) as a function of $\alpha$ in 8 scenarios of sparse binary or Poisson model with 10 active variables.}
    \label{fig:procedure-simcal-FDR-2}
\end{figure}

\begin{figure}
    \centering
    \makebox[\textwidth][c]{\includegraphics[width=\textwidth]{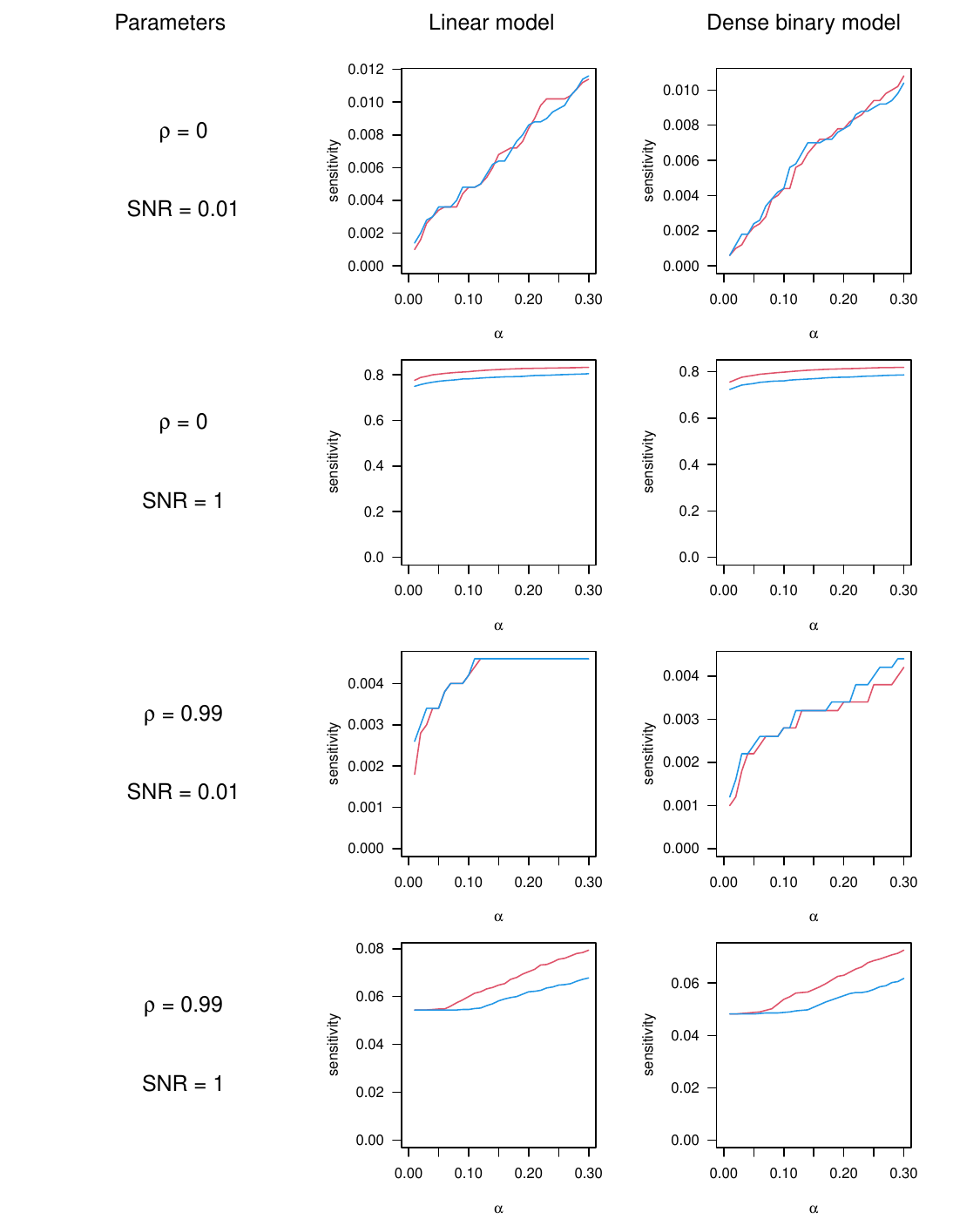}}
    \caption{Sensitivity of the procedure with thresholding (blue) or \textit{ForwardStop} (red) as a function of $\alpha$ in 8 scenarios of linear or dense binary model with 10 active variables.}
    \label{fig:procedure-simcal-sensitivity-1}
\end{figure}

\begin{figure}
    \centering
    \makebox[\textwidth][c]{\includegraphics[width=\textwidth]{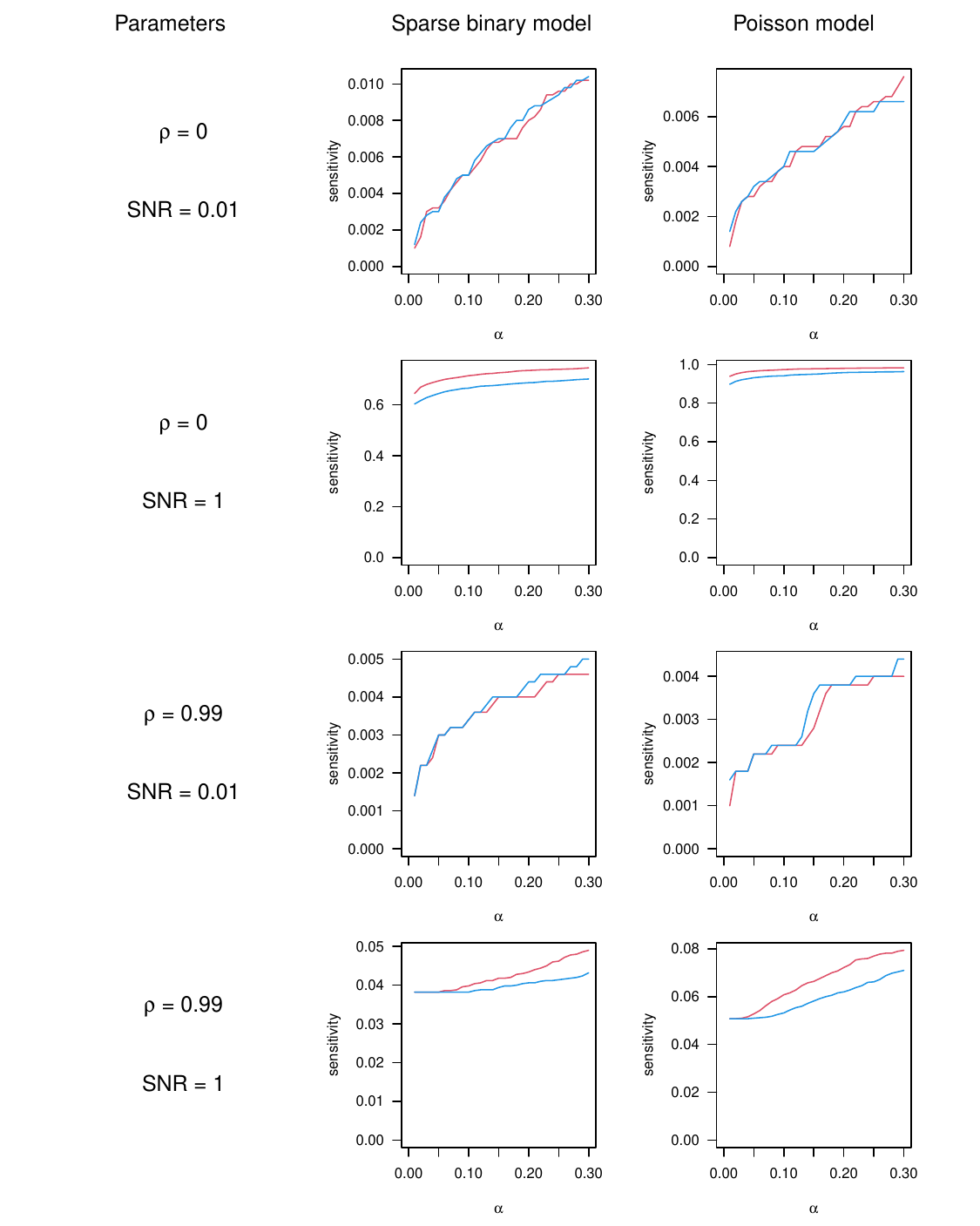}}
    \caption{Sensitivity of the procedure with thresholding (blue) or \textit{ForwardStop} (red) as a function of $\alpha$ in 8 scenarios of sparse binary or Poisson model with 10 active variables.}
    \label{fig:procedure-simcal-sensitivity-2}
\end{figure}

\begin{figure}
    \centering
    \makebox[\textwidth][c]{\includegraphics[width=0.9\textwidth]{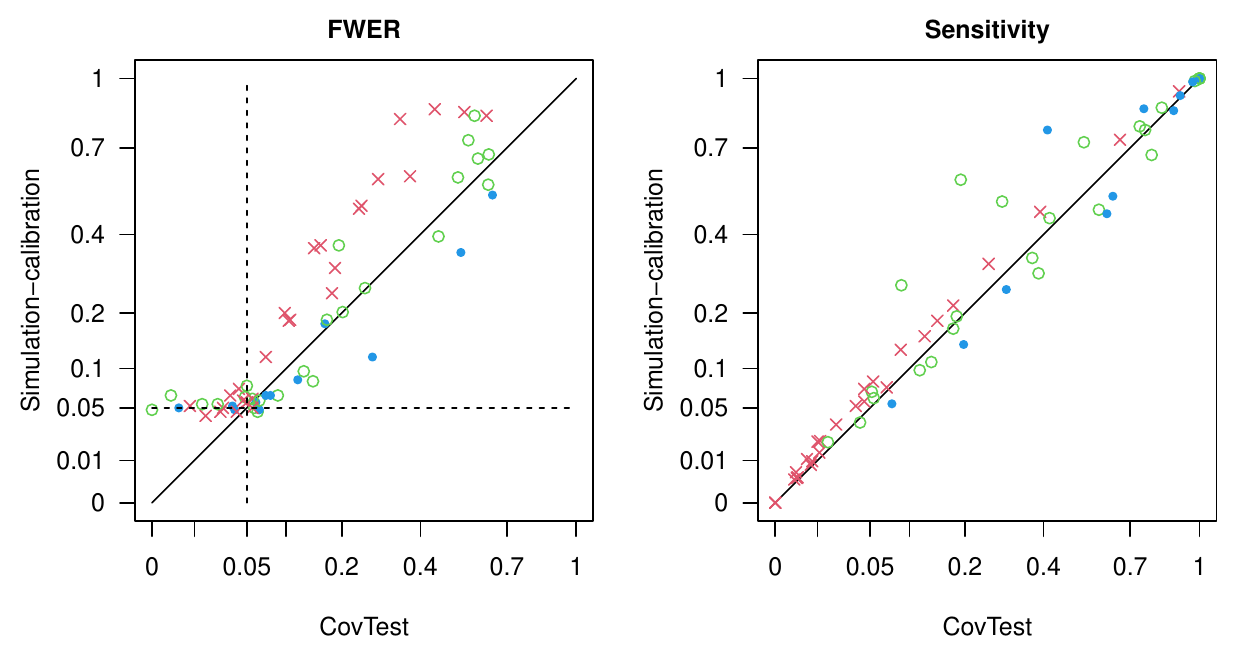}}
    \caption{Comparison of the performance of selection procedures with thresholding of CovTest and simulation-calibration test at $\alpha = 0.05$ across the 63 linear models (\textcolor{paletteBlue}{$\bullet$} $\mathrm{SNR} = 1$, \textcolor{paletteGreen}{$\circ$} $\mathrm{SNR} = 0.3$ or $0.1$, \textcolor{paletteRed}{$\times$} $\mathrm{SNR} = 0.03$, $0.01$ or  $0$). The position of a point indicates the performance of both methods on the same scenario, with the bisectors representing equal performance. The scale is quadratic.}
    \label{fig:comparaison-covtest-by-SNR}
\end{figure}

\begin{figure}
    \centering
    \makebox[\textwidth][c]{\includegraphics[width=0.9\textwidth]{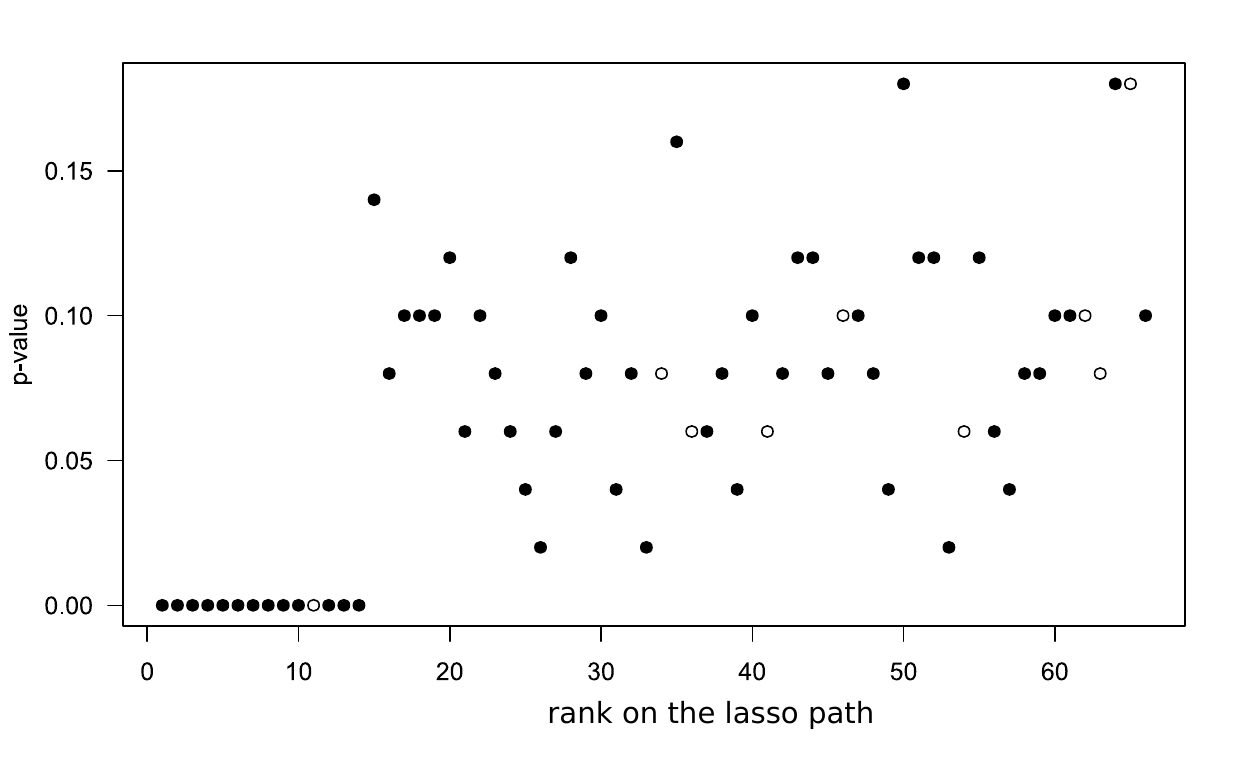}}
    \caption{Application to pharmacovigilance: p-values estimated by simulation-calibration along the Lasso path without removing the strongest correlations. $\bullet$: the association found is positive ($\widehat{\beta_j}^{Lasso} > 0$); $\circ$: negative.}
    \label{fig:simcal-BNPV-raw}
\end{figure}

\bibliography{biblio.bib}